\documentclass[runningheads,envcountsect,envcountsame,orivec,final]{llncs}

\usepackage[T1]{fontenc}
\usepackage{graphicx}
 
\usepackage{amsmath}
\usepackage{amsfonts}
\usepackage{amssymb}
\usepackage{galois}
\usepackage[citecolor=blue,linkcolor=blue]{hyperref}
\hypersetup{
    colorlinks,
    linkcolor={blue!65!black},
    citecolor={blue!65!black},
    urlcolor={blue!65!black}
}
\usepackage{stmaryrd}
\usepackage[dvipsnames]{xcolor}
\usepackage{bussproofs}
\usepackage{times}
\usepackage[nocompress]{cite}
\usepackage[textsize=tiny]{todonotes}
\usepackage{enumerate}
\usepackage{framed}
\usepackage{mathtools}

\spnewtheorem{definitionrm}[theorem]{Definition}{\bfseries}{\normalfont}
\spnewtheorem{examplebf}[theorem]{Example}{\bfseries}{\normalfont}

\usepackage{thm-restate}

\newcommand{\ok}{\textbf{\textcolor{Green}{ ok }}}
\newcommand{\err}{\textbf{\textcolor{red}{ err }}}
\newcommand{\green}[1]{\textcolor{Green}{#1}}
\newcommand{\red}[1]{\textcolor{red}{#1}}
\newcommand{\blue}[1]{\textcolor{blue}{#1}}

\newcommand{\disj} {
	\bigvee_{n \in \mathbb{N}}
}

\newcommand{\bdia}[1] {
	\boldsymbol{\langle} #1 \boldsymbol{|}
}

\newcommand{\intrp}[1] {
	\llbracket \mathtt{#1} \rrbracket_u
}

\newcommand{\csema}[1] {
	\llbracket \mathtt{#1} \rrbracket
}

\newcommand{\asema}[1] {
	\llbracket \mathtt{#1} \rrbracket_A^\sharp
}

\newcommand{\htri}[3] {
	\{#1\} \, #2 \, \{#3\}
}

\newcommand{\otri}[3] {
	[#1] \, #2 \, [#3]
}

\newcommand{\tri}[3] {
	 [#1] \; \mathtt{#2} \; [#3] 
}

\newcommand{\triok}[3] {
	 [#1] \; \mathtt{#2} \; \green{[\mathbf{ok} : #3]} 
}

\newcommand{\trierr}[3] {
	 [#1] \; \mathtt{#2} \; \red{[\mathbf{err} : #3]} 
}

\newcommand{\triokerr}[4] {
	 [#1] \; \mathtt{#2} \; \green{[\mathbf{ok} : #3]} \red{[\mathbf{err} : #4]}
}

\newcommand{\triboth}[3] {
	 [#1] \; \mathtt{#2} \; \blue{[\pmb{\epsilon} : #3]} 
}

\newcommand{\ptri}[3] {
	\vdash^{{\scriptscriptstyle\mathrm{K}}}_A [#1] \; \mathtt{#2} \; [#3] 
}

\newcommand{\ptritr}[3] {
	\vdash^{{\scriptscriptstyle\mathrm{K}}}_{A_{tr}} [#1] \; \mathtt{#2} \; [#3] 
}

\newcommand{\ptriul}[3] {
	\vdash_{\UL} [#1] \; \mathtt{#2} \; [#3] 
}

\newcommand{\ptriok}[3] {
	\vdash^{{\scriptscriptstyle\mathrm{K}}}_A [#1] \; \mathtt{#2} \; \green{[\mathbf{ok} : #3]} 
}

\newcommand{\ptriokil}[3] {
	\vdash_{\IL} [#1] \; \mathtt{#2} \; \green{[\mathbf{ok} : #3]} 
}

\newcommand{\ptrierr}[3] {
	\vdash^{{\scriptscriptstyle\mathrm{K}}}_A [#1] \; \mathtt{#2} \; \red{[\mathbf{err} : #3]} 
}

\newcommand{\ptrierrtr}[3] {
	\vdash^{{\scriptscriptstyle\mathrm{K}}}_{A_{tr}} [#1] \; \mathtt{#2} \; \red{[\mathbf{err} : #3]} 
}

\newcommand{\ptrierril}[3] {
	\vdash_{\IL} [#1] \; \mathtt{#2} \; \red{[\mathbf{err} : #3]} 
}

\newcommand{\ptriokerr}[4] {
	\vdash^{{\scriptscriptstyle\mathrm{K}}}_A [#1] \; \mathtt{#2} \; \green{[\mathbf{ok} : #3]} \red{[\mathbf{err} : #4]}
}

\newcommand{\ptriokerrtr}[4] {
	\vdash^{{\scriptscriptstyle\mathrm{K}}}_{A_{tr}} [#1] \; \mathtt{#2} \; \green{[\mathbf{ok} : #3]} \red{[\mathbf{err} : #4]}
}

\newcommand{\ptriokerril}[4] {
	\vdash_{\IL} [#1] \; \mathtt{#2} \; \green{[\mathbf{ok} : #3]} \red{[\mathbf{err} : #4]}
}

\newcommand{\ptriboth}[3] {
	\vdash^{{\scriptscriptstyle\mathrm{K}}}_A [#1] \; \mathtt{#2} \; \blue{[\pmb{\epsilon} : #3]} 
}

\newcommand{\ptribothil}[3] {
	\vdash_{\IL} [#1] \; \mathtt{#2} \; \blue{[\pmb{\epsilon} : #3]} 
}

\newcommand{\ptrit}[3] {
	\vdash^{{\scriptscriptstyle\mathrm{TK}}}_A [#1] \; \mathtt{#2} \; [#3] 
}

\newcommand{\ptrittr}[3] {
	\vdash^{{\scriptscriptstyle\mathrm{TK}}}_{A_{tr}} [#1] \; \mathtt{#2} \; [#3] 
}

\newcommand{\ptritokerr}[4] {
	\vdash^{{\scriptscriptstyle\mathrm{TK}}}_A [#1] \; \mathtt{#2} \; \green{[\mathbf{ok} : #3]} \red{[\mathbf{err} : #4]}
}

\newcommand{\ptritokerrtr}[4] {
	\vdash^{{\scriptscriptstyle\mathrm{TK}}}_{A_{tr}} [#1] \; \mathtt{#2} \; \green{[\mathbf{ok} : #3]} \red{[\mathbf{err} : #4]}
}

\newcommand{\ptritboth}[3] {
	\vdash^{{\scriptscriptstyle\mathrm{TK}}}_A [#1] \; \mathtt{#2} \; \blue{[\pmb{\epsilon} : #3]} 
}

\newcommand{\patriok}[4] {
	\vdash^{{\scriptscriptstyle\mathrm{K}}}_{#4} [#1] \; \mathtt{#2} \; \green{[\mathbf{ok} : #3]} 
}

\newcommand{\patrierr}[4] {
	\vdash^{{\scriptscriptstyle\mathrm{K}}}_{#4} [#1] \; \mathtt{#2} \; \red{[\mathbf{err} : #3]} 
}

\newcommand{\patrit}[4] {
	\vdash^{{\scriptscriptstyle\mathrm{TK}}}_{#4} [#1] \; \mathtt{#2} \; [#3] 
}

\newcommand{\vtri}[3] {
	\models^{{\scriptscriptstyle\mathrm{K}}}_A [#1] \; \mathtt{#2} \; [#3] 
}

\newcommand{\vtritr}[3] {
	\models^{{\scriptscriptstyle\mathrm{K}}}_{A_{tr}} [#1] \; \mathtt{#2} \; [#3] 
}

\newcommand{\vtriul}[3] {
	\models_{\UL} [#1] \; \mathtt{#2} \; [#3] 
}

\newcommand{\vtriok}[3] {
	\models^{{\scriptscriptstyle\mathrm{K}}}_A [#1] \; \mathtt{#2} \; \green{[\mathbf{ok} : #3]} 
}

\newcommand{\vtrierr}[3] {
	\models^{{\scriptscriptstyle\mathrm{K}}}_A [#1] \; \mathtt{#2} \; \red{[\mathbf{err} : #3]} 
}

\newcommand{\vtrierrtr}[3] {
	\models^{{\scriptscriptstyle\mathrm{K}}}_{A_{tr}} [#1] \; \mathtt{#2} \; \red{[\mathbf{err} : #3]} 
}

\newcommand{\vtriokerrtr}[4] {
	\models^{{\scriptscriptstyle\mathrm{K}}}_{A_{tr}} [#1] \; \mathtt{#2} \; \green{[\mathbf{ok} : #3]} \red{[\mathbf{err} : #4]}
}

\newcommand{\vtriokerril}[4] {
	\models_{\IL} [#1] \; \mathtt{#2} \; \green{[\mathbf{ok} : #3]} \red{[\mathbf{err} : #4]}
}

\newcommand{\vtrit}[3] {
	\models^{{\scriptscriptstyle\mathrm{TK}}}_A [#1] \; \mathtt{#2} \; [#3] 
}

\newcommand{\vtrittr}[3] {
	\models^{{\scriptscriptstyle\mathrm{TK}}}_{A_{tr}} [#1] \; \mathtt{#2} \; [#3] 
}

\newcommand{\ilc}[2]{\ok \green{: #1}, \err \red{: #2}}

\newcommand{\intrpok}[1] {
	\green{\llbracket \mathtt{#1} \rrbracket_{u_{\ok}}}
}

\newcommand{\intrperr}[1] {
	\red{\llbracket \mathtt{#1} \rrbracket_{u_{\err}}}
}

\newcommand{\asemaok}[1] {
	\green{\llbracket \mathtt{#1} \rrbracket^\sharp_{A_\ok}}
}
\newcommand{\asematrok}[1] {
	\green{\llbracket \mathtt{#1} \rrbracket^\sharp_{{A_{tr}}_\ok}}
}

\newcommand{\asemaerr}[1] {
	\red{\llbracket \mathtt{#1} \rrbracket^\sharp_{A_\err}}
}
\newcommand{\asematrerr}[1] {
	\red{\llbracket \mathtt{#1} \rrbracket^\sharp_{{A_{tr}}_\err}}
}

\newcommand{\csemaok}[1]{
	\green{\llbracket \mathtt{#1} \rrbracket_\ok}
}

\newcommand{\csemaerr}[1]{
	\red{\llbracket \mathtt{#1} \rrbracket_\err}
}

\newcommand{\udrshort}{\stackrel{{\mbox{\tiny\ensuremath{\vartriangle}}}}{\Leftrightarrow}}
\newcommand{\ra}{\rightarrow}
\newcommand{\ud}{\triangleq}
\newcommand{\tuple}[1]{\langle {#1}\rangle}

\DeclareMathOperator{\Test}{\mathsf{test}}

\DeclareMathOperator{\Atom}{Atom}
\DeclareMathOperator{\Error}{\mathtt{err}}

\DeclareMathOperator{\LCL}{LCL}
\DeclareMathOperator{\LCK}{LCK}
\DeclareMathOperator{\LCIL}{LCIL}
\DeclareMathOperator{\LCTK}{LCTK}
\DeclareMathOperator{\LCTIL}{LCTIL}

\DeclareMathOperator{\IL}{IL}
\DeclareMathOperator{\UL}{UL}
\DeclareMathOperator{\Int}{Int}
\DeclareMathOperator{\Sign}{Sign}
\newcommand{\bZ}{{\mathbb{Z}}}
\newcommand{\cG}{{\mathcal{G}}}
\DeclareMathOperator{\cod}{cod}
\DeclareMathOperator{\topp}{\mathsf{top}}
\DeclareMathOperator{\Spec}{Spec}

\newenvironment{scprooftree}[1]%
  {\gdef\scalefactor{#1}\begin{center}\proofSkipAmount \leavevmode}%
  {\scalebox{\scalefactor}{\DisplayProof}\proofSkipAmount \end{center} }

\begin{document}
\title{Local Completeness Logic on Kleene Algebra with Tests}
\pagestyle{plain}

\author{Marco Milanese \and Francesco Ranzato}
\authorrunning{M.\ Milanese and F.\ Ranzato}
\institute{Dipartimento di Matematica, University of Padova, Italy\\
}

\maketitle              %
\begin{abstract}
Local Completeness Logic (LCL) has been put forward as a program logic for proving both the correctness and incorrectness of program specifications. LCL is an abstract logic, parameterized by an abstract domain that allows combining over- and under-approximations of program behaviors. It turns out that LCL instantiated to the trivial singleton abstraction boils down to O'Hearn incorrectness logic, which allows us to prove the presence of program bugs. It has been recently proved that suitable extensions of Kleene algebra with tests (KAT) allow representing both O'Hearn incorrectness and Hoare correctness program logics within the same equational framework. In this work, we generalize this result by showing how KATs extended either with a modal diamond operator or with a top element are able to represent the local completeness logic LCL. This is achieved by studying how these extended KATs can be endowed with an abstract domain so as to define the validity of 
correctness/incorrectness LCL triples  and to show that the LCL proof system is logically sound and, under some hypotheses, 
complete.  
\end{abstract}
\section{Introduction}

Kleene algebra \cite{conway2012regular} with tests (KAT) \cite{kozen_kleene_1997} allows an
equational reasoning on programs and their properties. 
Programs are modeled as elements of a KAT, so that their properties can be algebraically derived  
through the general equational theory of KATs.  KATs feature sound, complete, and decidable equational theories
and have found successful applications 
in several different contexts, most notably in network programming~\cite{NetKAT,Cantor,BeckettGW16,FosterKM0T15,SmolkaEFG15}.  
The foundational study of Kozen~\cite{kozen_hoare_2000} has shown that the reasoning of 
Hoare correctness logic \cite{hoare_axiomatic_1969} can be encoded and formulated equationally within a KAT. Later work by Desharnais, M\"{o}ller and Struth \cite{desharnais_kleene_2006,moller_algebras_2006} extended KAT with a domain (KAD) to express the modal operators of propositional dynamic logic \cite{fischer_propositional_1979}, thus enabling a more natural way of reasoning through 
a map from actions to propositions. The expressive power of KAD has been recently substantiated by M\"{o}ller, O'Hearn and Hoare~\cite{fahrenberg_algebra_2021}, who have shown how to encode both Hoare~\cite{hoare_axiomatic_1969} correctness and O'Hearn~\cite{ohearn_incorrectness_2020} incorrectness program logics in a unique class of KAD where a backward diamond modality is exploited to encode strongest postconditions. 
Furthermore, very recently, Zhang, De Amorim and Gaboardi~\cite[Theorem~1]{zhang_incorrectness_2022}  have shown that O'Hearn incorrectness 
logic cannot be formulated within a conventional KAT, but, at the same time, a full fledged modal KAT is not needed. In fact, 
\cite{zhang_incorrectness_2022} proves that a KAT including a greatest element, called TopKAT, 
is capable to encode both Hoare and O'Hearn logic in a purely equational fashion. Moreover, \cite{zhang_incorrectness_2022} provides a PSPACE algorithm to decide TopKAT equality, based on a reduction to Cohen et al.~\cite{Cohen96thecomplexity}'s algorithm for KAT.
\\
\indent
This stream of works made it possible to reason equationally on both program correctness and incorrectness in the 
same algebraic framework. For example, in the KAD framework where a backward diamond modality $\bdia{a}p$ plays the role
of strongest postcondition of a KAT element $a$ (viz., a program) for a KAT test $p$ (viz., a precondition),  
the validity of 
a Hoare correctness triple $\htri{p}{a}{q}$ is determined by the 
inequality $\bdia{a} p \leq q$, while the validity of an O'Hearn incorrectness 
triple $\tri{p}{\text{$a$}}{q}$ boils down to $q\leq \bdia{a} p$. Moreover, if a KAT test $s$ plays the role of 
specification for a program $a$ and a Hoare triple $\htri{p}{\text{$a$}}{q}$ is provable, then $a$ can be proved correct through the inequality 
$q\leq s$. Vice versa, if $\tri{p}{\text{$a$}}{q}$ is a provable incorrectness triple, 
then incorrectness of $a$ can be verified as $q \leq \lnot s$.

\paragraph{\textbf{The Problem.}}
Recently, Bruni et al.~\cite{bruni_logic_2021} put forward a novel program logic, called local completeness logic $\LCL$, which 
is parameterized by an abstract domain \cite{CC77,CC79} of program stores and 
simultaneously combines over- and under-approximations of program behaviours. 
This program logic leverages the notion of \emph{locally complete} abstract interpretation, meaning that the abstract interpretation of atomic
program commands, such as variable assignments and Boolean guards, is complete (i.e.\ with no false alarm) \emph{locally} on the preconditions, as opposed
to standard completeness \cite{GRS00,Ranzato13} which must be satisfied \emph{globally} for all the preconditions. 
While a global completeness program logic was proposed in \cite{glr15},  
Bruni et al.~\cite{bruni_logic_2021}  design a proof system for inferring that a program analysis is locally complete.  
It turns out  \cite[Section~VI]{bruni_logic_2021} that the instantiation of
this $\LCL$ program logic to 
the trivial store abstraction with a unique ``don't know'' value abstracting any concrete store property, 
boils down to O'Hearn incorrectness logic \cite{ohearn_incorrectness_2020}.  Moreover, 
Bruni et al.~\cite{pldi2022} also show that abstract interpretations can be made locally complete through minimal 
domain refinements that repair the lack of local completeness in a given program analysis. 
\\
\indent
In the original definition of $\LCL$ in \cite{bruni_logic_2021} program properties are represented
as elements of a concrete domain $C$ and program semantics as functions of type $C \ra C$.
Although straightforward, this approach determines a specific type of program semantics.
Vice versa, by exploiting a KAT, program properties are represented as tests and programs
as generic elements of the KAT. Hence, a KAT based formulation becomes agnostic w.r.t.\ the
underlying
semantics and can therefore admit multiple different models of computation (e.g., trace-based semantics, or even
models not related to  program semantics as shown by the language-theoretic example in Section~\ref{ltKAT-subsec}). 
Furthermore, KAT is a particularly
suitable formalism for compositionally reasoning on programs as all its basic composition operations on
programs (concatenation, choice and Kleene iteration) are directly modeled within the algebra:
this
allows us to represent composite programs and tests as elements of the KAT and, in particular, to check
for their equality and inclusion directly in the algebra. Thus, 
following the KAT-based model of incorrectness logic 
advocated by M\"{o}ller, O'Hearn and Hoare~\cite{fahrenberg_algebra_2021}, this paper 
pushes forward this line of work by studying an algebraic
formulation of $\LCL$ program logic, with the objective of showing 
that there is no need to leverage particular semantic 
properties of programs to reason on their local completeness.

\paragraph{\textbf{Contributions.}}
In this work,  we show that  the local completeness logic $\LCL$ can be made fully algebraic in a suitable KAT, yet 
preserving all its noteworthy logical properties proved in \cite{bruni_logic_2021}.
For this purpose we show that:
\begin{itemize}
\item Our proof systems are logically sound and complete (likewise \cite{bruni_logic_2021}, 
completeness needs some additional
hypotheses).
\item By instantiating the algebraic version of the $\LCL$ logic to the trivial domain abstracting any
concrete value to ``don't know'' we exactly obtain O'Hearn incorrectness 
program logic \cite{ohearn_incorrectness_2020}, thus retrieving its logical soundness and completeness as consequences of our results. 
\item Triples of O'Hearn incorrectness logic carry two postconditions, corresponding to normal and erroneous program termination. 
While the original local completeness logic $\LCL$ in \cite{bruni_logic_2021} only considers normal termination, we propose a generalization that also 
supports erroneous termination. Moreover, we use the KAD construction of \cite{fahrenberg_algebra_2021} 
to generalize our logical soundness and completeness results to incorrectness triples.

\end{itemize}

In particular, we study two different formulations of $\LCL$ given: (1)~in a KAD, the KAT model used 
in \cite{fahrenberg_algebra_2021}, and (2)~in a TopKAT, the KAT model employed in \cite{zhang_incorrectness_2022}. 
In both frameworks, we put forward a suitable notion of abstract domain of KAT that, correspondingly, induces a  
sound abstract semantics for KAT programs (i.e., KAT terms). Our local completeness logic on KAT, called $\LCK$, turns out to be logically sound 
w.r.t.\ this abstract semantics, meaning that a provable 
$\LCK$ triple  $\tri{p}{\text{$a$}}{q}$ for an abstract domain $A$ on a KAT $K$ satisfies: 
\begin{enumerate}[{\rm (i)}]
\item $q$ is below the strongest postcondition in $K$ of the program term $a$ for the precondition $p$;
\item the program term $a$ is locally complete for the precondition $p$  in the abstract domain $A$;
\item the approximations in $A$ of $q$ and of the strongest postcondition of $a$ for $p$ coincide.
\end{enumerate}

\section{Background on Kleene Algebra with Tests}
A Kleene algebra with tests (KAT) is a purely algebraic structure that provides 
an elegant equational framework for program reasoning.
A KAT consists of actions, playing the role of 
programs, and tests, interpreted as pre/postconditions and Boolean guards. KAT elements can be combined with three 
basic operations: nondeterministic choice  $a_1+a_2$, sequential composition $a_1;a_2$, and Kleene iteration $a^*$. 
A standard model of KAT used to represent computations is the relational model, in which KAT elements are binary 
relations on some set, thus modeling programs as a relation between input and output states. 
Further models of KAT include regular languages over a finite alphabet,  
square matrices over another Kleene algebra, and Kleene algebra modulo theories~\cite{Greenberg22}. 
In the following, we briefly recall some basics of KAT. For more details, the reader is referred to \cite{conway2012regular,kozen_kleene_1997,desharnais_kleene_2006}.  

An \emph{idempotent-semiring} (\emph{i-semiring}) is a tuple $(A,+,\cdot,0,1)$ where:
(1)~$(A,+,0)$ is a commutative monoid with an idempotent addition, i.e., for all $a\in A$, $a+a=a$;
(2)~$(A,\cdot,1)$ is a monoid, where the multiplication symbol $\cdot$ is often omitted, such that,
for any $a \in A$, $0 \cdot a = a \cdot 0$;
(3)~multiplication distributes over addition (in both arguments).
In an i-semiring $A$, the relation $a \leq b \udrshort a + b = b$ is a partial order, referred to 
as the natural ordering, that we will implicitly use throughout the paper. Note that the addition 
$+$ is the join w.r.t.\ this natural ordering.

A \emph{test-semiring} is a tuple 
$(A,\Test(A),+,\cdot,\lnot,0,1)$ where: (1)~$(A,+,\cdot,0,1)$ is an i-semiring;  (2)~$\Test(A)\subseteq A$, 
and $(\Test(A), \lor, \land, \lnot, 0, 1)$ is a Boolean subalgebra of $A$ with greatest element $1$ and least element $0$, complement $\lnot$, where the meet $\land$ and join $\lor$ of the Boolean algebra $\Test(A)$
coincide, resp., with multiplication $\cdot$ and addition $+$.

A \emph{Kleene algebra} is a tuple $(K,+,\cdot,\mathord{^*},0,1)$ where: 
(1)~$(K,+,\cdot,0,1)$ is an i-semi\-ring; (2)~$(\cdot)^*:K\ra K$ is a unary operation, called Kleene star or iteration, 
satisfying the following conditions: 
	\begin{align*}
		\tag{$*$-unfold} \label{eqn:star_unfold} &1 + aa^* \leq a^* && 1 + a^*a \leq a^* \\
		\tag{$*$-induction} &b + ac \leq c \Rightarrow a^*b \leq c && b + ca \leq c \Rightarrow ba^* \leq c
	\end{align*}

\begin{definitionrm}[KAT \cite{kozen_kleene_1997}]
	A \emph{Kleene algebra with tests} (\emph{KAT}) is a two-sorted algebra
	$(K,\Test(K),$ $+, \cdot,\mathord{^*},$ $\lnot,0,1)$ such that
	$(K,\Test(K),+,\cdot,\lnot,0,1)$ is a test-semiring and $(K,+,\cdot,\mathord{^*},0,1)$ is a Kleene algebra.
	
	\noindent
	A KAT $K$ is \emph{countably-test-complete} (\emph{CTC}) if any countable subset of $\Test(K)$ admits least upper bound (lub).
	
	\noindent
	A KAT is \textit{$^*$-continuous}, referred to as KAT$^*$, if it satisfies the following condition: for all $a,b,c\in K$, 
	\( a b^* c = \bigvee_{n\in\mathbb{N}} a b^n c \) (this equation implicitly assumes that 
	the lub $\bigvee_{n\in\mathbb{N}} a b^n c$, w.r.t.\ the natural ordering of $K$,  exists).    
\end{definitionrm}

A \emph{relational KAT} 
\cite{goos_kleene_1997} on a carrier set $X$ is determined by a set $K\subseteq \wp(X\times X)$ of binary 
relations on $X$ with
tests $\Test(K)\subseteq \wp(\{(x,x) \mid x \in X\})$, where addition is union, multiplication is %
composition of relations, 
the additive identity is the empty relation, the multiplicative identity is $\{(x,x) \mid x \in X\}$, the Kleene star is the reflexive-transitive closure, and test complement is set complementation w.r.t.\ the multiplicative identity. 

Informally, a backward diamond $\bdia{\,\mathord{\cdot}\,}\,\cdot\:$
on a KAT allows us to compute strongest postcon\-ditions of programs, that is, $\bdia{a} p$ can be interpreted as $\mathsf{post}[a] p$.
\begin{definitionrm}[bdKAT \cite{fahrenberg_algebra_2021}]
	A \emph{backward-diamond KAT} (\emph{bdKAT}) is a two-sorted algebra
	$(K,$ $\Test(K),+,\cdot,\mathord{^*},\lnot,0,1,\bdia{})$ such that: 
	\begin{enumerate}[{\rm (1)}]
	\item $(K,\Test(K),+,\cdot,\mathord{^*},\lnot,0,1)$ is a KAT; 
	\item $\bdia{\,\mathord{\cdot}\,}\,\cdot : K \rightarrow (\Test(K) \rightarrow \Test(K))$ is a \emph{backward-diamond} operator 
	satisfying the following conditions: for all $a,b\in K$ and $p,q\in \Test(K)$, 
	\begin{align*}
		\tag{bd1} \label{eqn:bdia1} \bdia{a} p \leq q \Leftrightarrow pa \leq aq \\
		\tag{bd2} \label{eqn:bdia2} \bdia{a b} p = \bdia{b} (\bdia{a} p) %
	\end{align*}\qed
	\end{enumerate}
\end{definitionrm}

The following proposition summarizes few well-known operative properties of the ${\bdia{\cdot}}$ operator. Such operator can be defined using the codomain $\rho(\cdot)$ as $\bdia{a} p = \rho(pa)$ (e.g., as  in \cite[Section 5]{moller_algebras_2006}) and those properties are simply derived as a consequence of \cite[Lemma 4.11]{desharnais_kleene_2006}. On the other hand, in this paper we consider a (equivalent) direct axiomatization of the backward-diamond; therefore for the sake of completeness we provide a direct proof.
\begin{proposition} \label{prop:prop_bdia}
	Let $(K,\Test(K),+,\cdot,^*,\lnot,0,1,\bdia{})$ a backward-diamond KAT. For any $a,b \in K$ and $p,q \in \Test(K)$, the following properties hold
	\begin{align}
		\label{eqn:bdia_add_cmd} \bdia{a+b} p = \bdia{a} p + \bdia{b} p \\
		\label{eqn:bdia_add_cond} \bdia{a} (p+q) = \bdia{a} p + \bdia{a} q \\
		\label{eqn:bdia_isotony_cmd} a \leq b \Rightarrow \bdia{a} p \leq \bdia{b} p \\
		\label{eqn:bdia_isotony_cond} p \leq q \Rightarrow \bdia{a} p \leq \bdia{a} q \\
		\label{eqn:bdia_test} \bdia{s} p = p \cdot s \\
		\label{eqn:bdia_star_unfold} p + \bdia{a} \bdia{a^*} p \leq \bdia{a^*} p
	\end{align}
	If in addition, $K$ satisfies the CTC condition,
	\begin{equation} \label{eqn:bdia_star}
		\bdia{a^*} p = \disj \bdia{a^n} p.
	\end{equation}
\end{proposition}
\begin{proof}
	\leavevmode
	
	\noindent (\ref{eqn:bdia_add_cmd}): For any $p,q \in \Test(K)$ and $a \in K$
	\begin{align*} 
		pa \leq aq &\Rightarrow paq' \leq aqq' \Rightarrow paq' = 0 \\
		paq' = 0 &\Rightarrow pa = pa \cdot 1 = pa (q + q') = paq + paq' = paq \Rightarrow pa \leq aq
	\end{align*}
	thus
	\begin{equation} \label{eqn:claim_9}
		pa \leq aq \Leftrightarrow paq' = 0
	\end{equation}
	where $\lnot q$ is denoted with $q'$. We have
	\begin{align*}
		\bdia{a+b} p \leq q &\Leftrightarrow \textrm{[By (\ref{eqn:bdia1})]} \\
		p(a+b) \leq (a+b)q &\Leftrightarrow \textrm{[By (\ref{eqn:claim_9})]} \\
		p(a+b) q' = 0 &\Leftrightarrow \\
		paq' + pbq' = 0 &\Leftrightarrow \\
		paq' = 0 \land pbq' = 0 &\Leftrightarrow \textrm{[By (\ref{eqn:claim_9})]} \\
		pa \leq aq \land pb \leq bq &\Leftrightarrow \textrm{[By (\ref{eqn:bdia1})]} \\
		\bdia{a} p \leq q \land \bdia{b} p \leq q &\Leftrightarrow \\
		\bdia{a} p + \bdia{b} p \leq q
	\end{align*}
	Finally, since $q$ is arbitrary:
	\begin{align*}
		\bdia{a+b} p \leq \bdia{a+b} p \Rightarrow \bdia{a} p + \bdia{b} p \leq \bdia{a+b} p \\
		\bdia{a} p + \bdia{b} p \leq \bdia{a} p + \bdia{b} p \Rightarrow \bdia{a+b} p \leq \bdia{a} p + \bdia{b} p \\
	\end{align*}
	implying the claim $\bdia{a+b} p = \bdia{a} p + \bdia{b} p$.
	
	\vspace{5px} \noindent (\ref{eqn:bdia_add_cond}): Same as (\ref{eqn:bdia_add_cmd}).
	
	\vspace{5px} \noindent (\ref{eqn:bdia_isotony_cmd}): By definition $a \leq b \Leftrightarrow a + b = b$, thus
	\[ \bdia{b} p = \bdia{a+b} p \stackrel{(\ref{eqn:bdia_add_cmd})}{=} \bdia{a} p + \bdia{b} p \]
	implying $\bdia{a} p \leq \bdia{b} p$.
	
	\vspace{5px} \noindent (\ref{eqn:bdia_isotony_cond}): Same as (\ref{eqn:bdia_isotony_cmd}).
	
	\vspace{5px} \noindent (\ref{eqn:bdia_test}): Since $p,s \in \Test(K)$ and $\Test(K)$ forms a boolean algebra $ps = sps$ and by (\ref{eqn:bdia1}) $\bdia{s} p \leq ps$. Still by (\ref{eqn:bdia1}), $ps \leq s \bdia{s} p \leq \bdia{s} p$, implying $\bdia{s} p = ps$.
	
	\vspace{5px} \noindent (\ref{eqn:bdia_star_unfold}): The axiom (\ref{eqn:star_unfold}) implies $1 + a^* a \leq a^*$ so that
	\begin{align*}
		1 + a^* a \leq a^* &\Rightarrow \textrm{[By (\ref{eqn:bdia_isotony_cmd})]} \\
		\bdia{1 + a^* a} p \leq \bdia{a^*} p &\Rightarrow \textrm{[By (\ref{eqn:bdia_add_cmd})]} \\
		\bdia{1} p + \bdia{a^* a} p \leq \bdia{a^*} p &\Rightarrow \textrm{[By (\ref{eqn:bdia_test})]} \\
		p + \bdia{a^* a} p \leq \bdia{a^*} p &\Rightarrow \textrm{[By (\ref{eqn:bdia2})]} \\
		p + \bdia{a} \bdia{a^*} p \leq \bdia{a^*} p.
	\end{align*}

	\vspace{5px} \noindent (\ref{eqn:bdia_star}): The proof of this property can be found in \cite[Lemma 3.3]{fahrenberg_algebra_2021}.	\qed
\end{proof}

The axiom \eqref{eqn:bdia1} is equivalent to requiring that $\bdia{a} p$ is the least test in $K$ satisfying $pa \leq aq$ (the original definition of Kleene algebra with domain in \cite{desharnais_kleene_2006}  is of this form). Moreover, $pa \leq aq$ in \eqref{eqn:bdia1} is equivalent to $pa = paq$ (see \cite[Lemma 3.4]{desharnais_kleene_2006}).
For the sake of completeness we provide a direct proof of this fact.

\begin{restatable}{lemma}{propbdiachar}\label{prop:bdia_char}
	Let $(K,\Test(K),+,\cdot,\mathord{^*},\lnot,0,1,\bdia{})$ a bdKAT. For all $p \in \Test(K)$ and $a \in K$,  
	\eqref{eqn:bdia1} holds iff $\bdia{a} p$ is the least (w.r.t.\ $\leq$) ${q \in \Test(K)}$ such that $pa = paq$.
\end{restatable}
\begin{proof}
	First we show the equivalence between $pa \leq aq$ and $pa = paq$. Indeed: ${pa \leq aq}$ implies $pa = ppa \leq paq$ and $1 \geq q$ implies $pa \geq paq$, meaning that $pa = paq$. Viceversa if $pa = paq$, then $p \leq 1$ implies $pa = paq \leq aq$.
	
	We can show that (\ref{eqn:bdia1}) is equivalent to say that $\bdia{a} p$ is the least $q \in \Test(K)$ satisfying ${pa \leq aq}$ (which is equivalent to $pa = paq$). In formulas the latter is
	\begin{align}
		\label{eqn:prop_rela_exts_claim_1} pa &\leq a \bdia{a} p \\
		\label{eqn:prop_rela_exts_claim_2} pa &\leq aq \Rightarrow \bdia{a} p \leq q
	\end{align}
	Assume (\ref{eqn:bdia1}). (\ref{eqn:prop_rela_exts_claim_2}) is the direction $\Leftarrow$ of (\ref{eqn:bdia1}). (\ref{eqn:prop_rela_exts_claim_1}) holds by instantiating $q = \bdia{a} p$ in (\ref{eqn:bdia1}).
	Viceversa, assume (\ref{eqn:prop_rela_exts_claim_1}) and (\ref{eqn:prop_rela_exts_claim_2}). (\ref{eqn:prop_rela_exts_claim_2}) is the direction $\Leftarrow$ of (\ref{eqn:bdia1}). The other direction holds because $pa \leq a \bdia{a} p \leq aq$. \qed
\end{proof}

\begin{definitionrm}[TopKAT \cite{esparza_equational_2017}]
A \emph{KAT with top} (\emph{TopKAT}) is a KAT $K$ that contains a largest element  $\top\in K$, 
that is, for all $a\in K$, $a\leq \top$. 
\qed
\end{definitionrm}

\section{Local Completeness Logic in KAT}\label{sec:lclkat}
We investigate how the local completeness program logic $\LCL$ 
\cite{bruni_logic_2021} can be interpreted on a KAT. To achieve this, we need to address
the following tasks: 
\begin{itemize}
	\item To define a notion of abstract domain of a KAT, with the aim of abstracting
the set of program predicates, namely tests of a KAT;
	\item To establish a concrete semantics and a corresponding sound abstract semantics of programs on KATs;
	\item To adapt the local completeness proof system to attain valid triples on a KAT; 
	\item To prove logical soundness and completeness w.r.t.\ a KAT of this new proof system.
\end{itemize}

\subsection{Program Properties in KAT} \label{sect:exts_prop_kat}
Program properties can be broadly classified as intensional and extensional. The former relate to how programs are written, while the latter concern the input-output relation of a program,  i.e., its strongest postcondition denotational semantics.
Local completeness logic $\LCL$ relies on 
an abstract interpretation of programs which crucially depends on intensional properties of programs, meaning
that even if two programs share the same denotation, they could well have different abstract semantics. 
Thus, we expect that an appropriate definition of abstract semantics based on a KAT model should also be intensional. 
Given two elements $a$ and $b$ of a modal bdKAT playing the role of programs, we therefore expect that 
their backward diamond functions might coincide, i.e.\ $\bdia{a} = \bdia{b}$, even if $a$ and $b$ encode different programs, i.e.\ $a\neq b$. 
However, as shown by the following remark for the basic
relational model of KAT, %
it might happen that for certain classes of KAT models  the backward diamond interpretation is injective. 

In order to prove such result, we require an auxiliary Lemma to provide an explicit formula for the $\bdia{\cdot}$ operator in the special case of full set of tests (i.e., $\Test(K)=\wp(\{(x,x) \mid x \in X\})$).

\begin{lemma} \label{lemma:bdia_rela}
	Let $K$ be a relational KAT on a set $X$ where $\Test(K)=\wp(\{(x,x) \mid x \in X\})$. Then, 
	for all $a \in K$ and $p \in \Test(K)$,
	\[ \bdia{a} p = \{(y,y) \mid \exists x \in X . \; (x,x) \in p, (x,y) \in a \} \]
\end{lemma}
\begin{proof}
	We want to show that $\bdia{a} p$ exists and is equal to
	\[ s = \{(y,y) \mid \exists x \in X . \; (x,x) \in p, (x,y) \in a \}. \]
	Using the characterization of $\bdia{a} p$ of Lemma~\ref{prop:bdia_char}, we need to show that $s$ is the least $q \in \Test(K)$ satisfying $pa = paq$. In this context
	\begin{align*}
		paq &= \{(x,y) \mid (x,x) \in p, (x,y) \in a, (y,y) \in q \} \\
		pa &= \{(x,y) \mid (x,x) \in p, (x,y) \in a \}
	\end{align*}
	
	First we show that $pa = pas$. Assume $(x,y) \in pa$, then by definition of $s$ it holds $(y,y) \in s$, meaning that $(x,y) \in pas$. This shows $pa \leq pas$, the other inequality is immediate since $s \leq 1$ so that $pa \geq pas$.
	
	Secondly we show that $s$ is the least such element. Assume by contradiction that $\exists t \in \Test(K)$ satisfying $pa = pat$ and $t < s$. Let $(y,y) \in s \setminus t$. By definition of $s$ exists $x \in X$ such that $(x,x) \in p$ and $(x,y) \in a$, therefore $(x,y) \notin pat$ because $t$ is a subidentity and $(y,y) \notin t$. This contradicts the hypothesis $pa = pat$. \qed
\end{proof}

\begin{restatable}{proposition}{proprelaexts} \label{prop:rela_exts}
	Let $K$ be a relational KAT on a set $X$ where $\Test(K)=\wp(\{(x,x) \mid x \in X\})$. Then, 
	for all $a,b \in K$, 
	\( \bdia{a} = \bdia{b} \; \Leftrightarrow \; a = b \).
\end{restatable}
\begin{proof}
	The direction $\Leftarrow$ is trivial, because $\bdia{\cdot}$ is a function on the first argument. Instead, we focus on the $\Rightarrow$ direction.
	
	Assume $\bdia{a} = \bdia{b}$. If $(x,y) \in a$, then by Lemma~\ref{lemma:bdia_rela}
	\[ \bdia{a} \{(x,x)\} = \bdia{b} \{(x,x)\} \supseteq \{(y,y)\} \]
	This implies the existence of $z \in X$ such that $(z,z) \in \{(x,x)\}$ and $(z,y) \in b$, but the only possibility is $z=x$ so that $(x,y) \in b$. By symmetry we can obtain $(x,y) \in b \Rightarrow (x,y) \in a$.
	
	Notice that we can be sure that $\{(x,x)\} \in \Test(K)$ because of the hypothesis $\Test(K) = \wp(\{(x,x) \mid x \in X\})$. \qed
\end{proof}

This means that, at least for some fundamental KAT models, KAT elements are equal iff they are extensionally
equal, or, equivalently, they carry exclusively extensional program properties. In this case,
when a program is encoded with a KAT element all the intensional properties are lost and it is indistinguishable
from any other program with the same denotational semantics. Therefore, an abstract interpretation-based
semantics can not be defined directly on KAT elements.

\subsection{KAT Language}
As a consequence of the discussion in Section~\ref{sect:exts_prop_kat}, the concrete semantics cannot be directly 
defined on KAT elements. A solution is to define it on an inductive language. Actually, in a language of programs, two elements are equal iff they are syntactically equal, or, in other terms, if the corresponding programs are written in the same way. This property makes a language an ideal basis upon which a semantics can be defined, because this brings the chance of depending on   
intensional properties.

A natural choice for defining this language of programs is the so-called \emph{KAT language}, as originally defined by Kozen and Smith~\cite[Section~2.3]{goos_kleene_1997}, because it contains all and only the operators of a KAT, so that the interpretation of language terms as 
KAT elements is the most natural one. This language 
is inductively defined from two disjoint sets of primitive actions and tests through the basic elements/operations $0,1,+,\cdot,^*$ of KATs. 
More precisely, given a set $\Sigma$ of \textit{primitive actions} and a set $B$ of  \textit{primitive tests} such that $\Sigma\cap B=\varnothing$, the corresponding \emph{KAT language}  $T_{\Sigma,B}$ of terms is defined as follows: 
\begin{align*}
&\Atom \ni \mathtt{a} ::= a \in \mathit{\Sigma} \mid p \in \mathit{B} \\
& T_{\Sigma,B}\ni \mathtt{t ::= \mathtt{a} \mid 0 \mid 1 \mid t_1 + t_2 \mid t_1 \cdot t_2 \mid t^*}
\end{align*}
For simplicity, we assume that $\mathtt{0}$ and $\mathtt{1}$ are primitive tests in $B$, 
so that $\mathtt{0},\mathtt{1}\in \Atom$. 
The notation $\Atom(\mathtt{t})\subseteq \Atom$ will denote the set of atoms occurring in a term 
$\mathtt{t} \in T_{\Sigma, B}$.
Notice that a KAT language $T_{\Sigma,B}$ is an equivalent representation of the language of regular commands 
used in \cite{ohearn_incorrectness_2020,bruni_logic_2021} for their program logics. 

Given a KAT $K$, an evaluation of atoms in $K$ is a mapping $u : \Atom \rightarrow K$ 
such that $\mathtt{p} \in B \Rightarrow u(\mathtt{p}) \in \Test(K)$.
An evaluation $u$ induces an interpretation of terms $\intrp{\cdot} : T_{\Sigma,B} \rightarrow K$, which is inductively 
defined as expected:
	\begin{align*}
		&\intrp{a} \ud u(\mathtt{a}) && 
		\intrp{t_1 + t_2} \ud \intrp{t_1} + \intrp{t_2}\\
		&\intrp{t_1 \cdot t_2} \ud \intrp{t_1} \cdot \intrp{t_2} &&
		\intrp{t^*} \ud \intrp{t}^*
	\end{align*}

In turn, the concrete semantic function
\[ \csema{\cdot}^K : T_{\Sigma,B} \rightarrow \big(\Test(K) \rightarrow \Test(K)\big) \]
models the strongest postcondition of a program, i.e.\ of a language term, for a 
given precondition, i.e.\ a KAT test. This is therefore defined in terms of the backward diamond of a bdKAT as follows:  
\begin{align}
 \csema{t}^K p \triangleq \bdia{\intrp{t}} p. \label{bd-csem}
\end{align}
We will often use $\csema{t}$ to denote a concrete semantics, by omitting the superscript $K$ 
when it is clear from the context. 

\noindent The construction of the concrete semantic leverages the backward-diamond (\ref{bd-csem}) and the interpretation $\intrp{\cdot}$, hence the properties of Proposition \ref{prop:prop_bdia} can be transferred to $\csema{\cdot}^K$. For completeness we summarize them in the following:
\begin{align}
	\label{eqn:csema_isotony} p \leq q \Rightarrow \csema{t} p \leq \csema{t} q \\
	\label{eqn:csema_add_cmd} \csema{t_1 + t_2} p = \csema{t_1} p + \csema{t_2} p \\
	\label{eqn:csema_add_cond} \csema{t} (p+q) = \csema{t} p + \csema{t} q \\
	\label{eqn:csema_mul} \csema{t_1 \cdot t_2} p = \csema{t_2} \csema{t_1} p \\
	\label{eqn:csema_star} \csema{t^*} p = \disj (\csema{t})^n p \\
	\label{eqn:csema_star_unfold} p + \csema{t} \csema{t^*} p \leq \csema{t^*} p
\end{align}

\subsection{Kleene Abstractions}
An abstract domain is used in abstract interpretation for approximating store properties, i.e., sets of program 
stores form the concrete domain, likewise in our KAT model, the role of concrete domain is played by the set of tests $\Test(K)$ of a KAT $K$, ordered by the natural 
ordering induced by $K$.  

\begin{definitionrm}[Kleene Abstract Domain]\label{def-kad}
	A poset $(A, \leq_A)$ is a \emph{Kleene abstract domain} of a bdKAT $K$ if: 
	\begin{enumerate}[{\rm (i)}]
	\item There exists a Galois insertion, defined by 
	a concretization map $\gamma : A \rightarrow \Test(K)$ and an abstraction map $\alpha : \Test(K) \rightarrow A$,  
		of the poset $(A,\leq_A)$ into the poset 
		${(\Test(K),\leq_K)}$;  
     \item $A$ is countably-complete, i.e., any countable subset of $A$ admits a lub. \qed
     \end{enumerate}
\end{definitionrm}

The abstract semantic function
\( \asema{\cdot} : T_{\Sigma, B} \rightarrow (A \rightarrow A) \)
defines how abstract preconditions are transformed into abstract postconditions. Likewise store-based abstract interpretation, 
this abstract semantics is inductively defined as follows: 
\begin{equation}
\begin{aligned}\label{def:abs-sem}
&\asema{a} p^\sharp \triangleq \alpha(\csema{a}^K \gamma(p^\sharp)) &&
	\quad\asema{t_{\text{$1$}} + t_{\text{$2$}} } p^\sharp \triangleq \asema{t_{\text{$1$}}} p^\sharp + \asema{t_{\text{$2$}}} p^\sharp \\
	&\asema{t_{\text{$1$}} \cdot t_{\text{$2$}}}p^\sharp \triangleq \asema{t_{\text{$2$}}}(\asema{t_{\text{$1$}}} p^\sharp) &&
	\quad\asema{t^*} p^\sharp \triangleq \textstyle \disj (\asema{t})^n p^\sharp
\end{aligned}
\end{equation}

We recall in the following Proposition few properties of Galois Insertion. Proofs and a more in depth-presentation are available in \cite{mine_tutorial_2017}.
\begin{proposition} \label{prop:prop_ai}
	Let $(C, \leq_C) \galois{\alpha}{\gamma} (A, \leq_A)$ a Galois Insertion.
	\begin{align}
		\label{eqn:alpha_isotony} \forall c_1,c_2 \in C. \; &c_1 \leq_C c_2 \Rightarrow \alpha(c_1) \leq_A \alpha(c_2) \\
		\label{eqn:gamma_isotony} \forall a_1, a_2 \in A. \; &a_1 \leq_A a_2 \Rightarrow \gamma(a_1) \leq_C \gamma(a_2) \\
		\label{eqn:gamma_inject} \forall a_1, a_2 \in A. \; &a_1 \neq a_2 \Rightarrow \gamma(a_1) \neq \gamma(a_2) \\
		\label{eqn:gamma_alpha} \forall c \in C. \; &c \leq_C \gamma(\alpha(c)) \\
		\label{eqn:uco_idem} &\gamma \comp \alpha \comp \gamma \comp \alpha = \gamma \comp \alpha
	\end{align}
\end{proposition}

It is worth remarking that condition (ii) of Definition~\ref{def-kad} ensures that the abstract semantics of the 
Kleene star  in~\eqref{def:abs-sem} is well defined. 
It turns out that $\asema{\cdot}$ is a sound (and monotonic) abstract semantics.  
\begin{restatable}[Soundness of bdKAT Abstract Semantics]{theorem}{propasema}\label{prop:prop_asema}
Let $A$ be a Kleene abstraction of a CTC bdKAT $K$ and $T_{\Sigma,B}$ be a language interpreted on $K$. 
	For all $p^\sharp, q^\sharp \in A$, $p\in \Test(K)$, and $\mathtt{t} \in T_{\Sigma, B}$:
	\begin{align*}
		\tag{monotonicity} &p^\sharp \leq_A q^\sharp \Rightarrow \asema{t} p^\sharp \leq_A \asema{t} q^\sharp \\
		\tag{soundness} &\alpha(\csema{t}^K p) \leq_A \asema{t} \alpha(p)
	\end{align*}
\end{restatable}
\begin{proof}
	Let us prove the first implication. The proof is by induction on the structure of $\mathtt{t} \in T_{\Sigma, B}$. \vspace{5px} \\
	($\mathtt{a} \in \Atom$)\,:
	\begin{align*}
		\asema{a} p^\sharp &= \\
		\alpha(\csema{a} \gamma(p^\sharp)) &\leq \textrm{[isotony of $\alpha(\cdot)$,$\gamma(\cdot)$,$\csema{\cdot}$]} \\
		\alpha(\csema{a} \gamma(q^\sharp)) &= \\
		\asema{a} q^\sharp 
	\end{align*}
	($\mathtt{t_1 + t_2}$)\,: By induction it holds $\asema{t_i} p^\sharp \leq \asema{t_i} q^\sharp$ for $i \in \{1,2\}$
	\[ \asema{t_1 + t_2} p^\sharp = \asema{t_1} p^\sharp + \asema{t_2} p^\sharp \leq \asema{t_1} q^\sharp + \asema{t_2} q^\sharp = \asema{t_1 + t_2} q^\sharp \ \]
	($\mathtt{t_1 \cdot t_2}$)\,: By induction it holds $\asema{t_1} p^\sharp \leq \asema{t_1} q^\sharp$ and also $\asema{t_2} (\asema{t_1} p^\sharp) \leq \asema{t_2} (\asema{t_1} q^\sharp)$
	\[ \asema{t_1 \cdot t_2} p^\sharp = \asema{t_2} \asema{t_1} p^\sharp \leq \asema{t_2} \asema{t_1} q^\sharp = \asema{t_1 \cdot t_2} q^\sharp \]
	($\mathtt{t_0^*}$)\,: The result is a consequence of the following claim
	\begin{equation} \label{eqn:claim_10}
		\forall n \in \mathbb{N}. \; p^\sharp \leq q^\sharp \Rightarrow (\asema{t_0})^n p^\sharp \leq (\asema{t_0})^n q^\sharp
	\end{equation}
	which can be shown by induction on $n$. ${(\asema{t_0})^0 p^\sharp = p^\sharp \leq q^\sharp = (\asema{t_0})^0 q^\sharp}$ and for the inductive case: $(\asema{t_0})^{n+1} p^\sharp = \asema{t_0} (\asema{t_0})^n p^\sharp \leq \asema{t_0} (\asema{t_0})^n q^\sharp \leq (\asema{t_0})^{n+1} q^\sharp$.
	Therefore
	\[ \asema{t_0^*} p^\sharp = \disj (\asema{t_0})^n p^\sharp \stackrel{(\ref{eqn:claim_10})}{\leq} \disj (\asema{t_0})^n q^\sharp = \asema{t_0^*} q^\sharp \]

	\noindent Similarly, the soundness is proven by induction on the structure of $\mathtt{t} \in T_{\Sigma, B}$. \vspace{5px} \\
	($\mathtt{a} \in \Atom$)\,: By isotony of $\alpha(\cdot)$, $\csema{\cdot}$ and (\ref{eqn:gamma_alpha}) we have
	\[ \alpha(\csema{a} p) \leq \alpha(\csema{a} \gamma(\alpha(p))) = \asema{a} \alpha(p) \]
	($\mathtt{t_1 + t_2}$)\,: By additivity of $\csema{\cdot}$ and $\alpha(\cdot)$ we have
	\begin{align*}
		\alpha(\csema{t_1 + t_2} p) &= \alpha(\csema{t_1} p + \csema{t_2} p) \\
		&= \alpha(\csema{t_1} p) + \alpha(\csema{t_2} p) \\
		&\leq \asema{t_1} \alpha(p) + \asema{t_2} \alpha(p) \\
		&= \asema{t_1 + t_2} \alpha(p)
	\end{align*}
	($\mathtt{t_1 \cdot t_2}$)\,: By isotony of $\asema{\cdot}$ we have
	\begin{align*}
		\alpha(\csema{t_1 \cdot t_2} p) &= \alpha(\csema{t_2} \csema{t_1} p) \\
		&\leq \asema{t_2} \alpha(\csema{t_1} p) \\
		&\leq \asema{t_2} \asema{t_1} \alpha(p) \\
		&= \asema{t_1 \cdot t_2} \alpha(p)
	\end{align*}
	($\mathtt{t_0^*}$)\,: In order to show the result we need an auxiliary claim:
	\begin{equation} \label{eqn:claim_1}
		\alpha((\csema{t_0})^n p) \leq (\asema{t_0})^n \alpha(p)
	\end{equation}
	The proof is by induction on $n$. The base case holds as $\alpha((\csema{t_0})^0 p) = \alpha(p) = (\asema{t_0})^0 \alpha(p)$ and similarly the inductive case as $\alpha((\csema{t_0})^{n+1} p) = \alpha((\csema{t_0})^n \csema{t_0} p) \leq (\asema{t_0})^n \alpha(\csema{t_0} p) \leq (\asema{t_0})^n \asema{t_0} \alpha(p) = (\asema{t_0})^{n+1} \alpha(p)$. 
	Therefore we have
	\begin{align*}
		\alpha(\csema{t_0^*} p) &= \textrm{[By (\ref{eqn:csema_star})]} \\
		\alpha(\disj (\csema{t_0})^n p) &= \textrm{[By additivity of $\alpha(\cdot)$]} \\
		\disj \alpha((\csema{t_0})^n p) &\leq \textrm{[By (\ref{eqn:claim_1})]} \\
		\disj (\asema{t_0})^n \alpha(p) &= \\
		\asema{t_0^*} \alpha(p)
	\end{align*}
	Notice that the CTC condition on $K$ must hold in order to apply (\ref{eqn:csema_star}). \qed
\end{proof}

\subsection{Local Completeness Logic on bdKAT}
Given a Kleene abstract domain $A$, we will slightly abuse notation by using 
\begin{equation*}
A \ud \gamma\comp\alpha : \Test(K)\ra \Test(K)
\end{equation*} 
as a function (indeed, this is the upper closure operator on tests induced by the Galois insertion defining $A$). 
Let us recall the notions of global vs.\ local completeness.  If $f : \Test(K) \rightarrow \Test(K)$ is any test transformer  then:
		\begin{itemize}
		\item $A$ is \emph{globally complete} for $f$, denoted $\mathbb{C}^A(f)$, iff $A \comp f = A \comp f \comp A$;
		\item $A$ is \emph{locally complete} for $f$ on a concrete test $p \in \Test(K)$, denoted $\mathbb{C}^A_p(f)$, iff
		$A \comp f(p) = A \comp f \comp A(p)$.
\end{itemize}

It is known \cite{glr15} 
that global completeness is hard to achieve in practice, even for simple programs. 
Moreover, a complete and compositional (i.e., inductively defined on program structure)  
abstract interpretation is even harder to design \cite{BruniGGGP20}. This motivated to study a local notion of completeness
in abstract interpretation \cite{bruni_logic_2021} as a pragmatic and more attainable weakening of standard global completeness.

In our local completeness logic on a Kleene algebra, a triple $\tri{p}{t}{q}$, where $p$ and $q$ are 
tests and $\mathtt{t}$ is a language term, will be valid when: 
\begin{enumerate}[{\rm (1)}]
	\item $q$ is an under-approximation of the concrete semantics of $\mathtt{t}$ from a precondition $p$;
	\item $A$ is locally complete for $\csema{t}$ on the precondition $p$;
	\item $q$ and $\csema{t} p$ have the same over-approximation in $A$.
\end{enumerate}

\begin{definitionrm}[Triple Validity]\label{def:validity}
	Let $K$ a CTC bdKAT, $A$ be a Kleene abstraction of $K$, and $T_{\Sigma, B}$ be a KAT language interpreted on $K$.
	For all $p,q \in \Test(K)$ and $\mathtt{t} \in T_{\Sigma, B}$, a triple 
	$\tri{p}{t}{q}$ is valid in $A$, denoted by $\vtri{p}{t}{q}$, if
		\begin{enumerate}[{\rm (i)}]
			\item\label{val-cond1} $q \leq_K \csema{t}^K p$;
			\item\label{val-cond2} $\asema{t} \alpha(p) = \alpha(q) = \alpha(\csema{t}^K p)$. \qed
		\end{enumerate}
\end{definitionrm}

The local completeness proof system in \cite{bruni_logic_2021} can be adapted to our algebraic framework, 
yielding the set of rules denoted by $\LCK_A$ in Figure \ref{fig:lcl}.
The only syntactic difference concerns the usage of elements of $\Test(K)$ as pre/postconditions and the language of terms 
$T_{\Sigma, B}$ playing the role of programs.

\begin{figure}[t]
	\centering
	\begin{framed}
	\vspace*{-5pt}
	\begin{minipage}{\textwidth}
		\begin{minipage}{\textwidth}
			\begin{prooftree}
				\AxiomC{$\mathtt{a} \in \Sigma \cup B$ \qquad $\mathbb{C}^A_p(\csema{a})$}
				\RightLabel{(transfer)}
				\UnaryInfC{$\ptri{p}{a}{\csema{a} p}$}
			\end{prooftree}
		\end{minipage}
		\\[3pt]
		\begin{minipage}{\textwidth}
			\begin{prooftree}
				\AxiomC{$p' \leq p \leq\! A(p') \qquad \ptri{p'}{t}{q'} \qquad q \leq q' \leq\! A(q')$}
				\RightLabel{(relax)}
				\UnaryInfC{$\ptri{p}{t}{q}$}
			\end{prooftree}
		\end{minipage}
	\end{minipage}
	\\[3pt]
	\begin{minipage}{\textwidth}
		\begin{minipage}{0.5\textwidth}
			\begin{prooftree}
				\AxiomC{$\ptri{p}{t_{\text{$1$}}}{r}$}
				\AxiomC{$\ptri{r}{t_{\text{$2$}}}{q}$}
				\RightLabel{(seq)}
				\BinaryInfC{$\ptri{p}{t_{\text{$1$}} \cdot t_{\text{$2$}}}{q}$}
			\end{prooftree}
		\end{minipage}
		\hfill
		\begin{minipage}{0.5\textwidth}
			\begin{prooftree}
				\AxiomC{$\ptri{p}{t_{\text{$1$}}}{q_1}$}
				\AxiomC{$\ptri{p}{t_{\text{$2$}}}{q_2}$}
				\RightLabel{(join)}
				\BinaryInfC{$\ptri{p}{t_{\text{$1$}} + t_{\text{$2$}}}{q_1 + q_2}$}
			\end{prooftree}
		\end{minipage}
	\end{minipage}
    \\[3pt]
	\begin{minipage}{\textwidth}
		\begin{minipage}{0.5\textwidth}
			\begin{prooftree}
				\AxiomC{$\ptri{p}{t}{r}$}
				\AxiomC{$\ptri{p + r}{t^*}{q}$}
				\RightLabel{(rec)}
				\BinaryInfC{$\ptri{p}{t^*}{q}$}
			\end{prooftree}
		\end{minipage}
		\hfill
		\begin{minipage}{0.5\textwidth}
			\begin{prooftree}
				\AxiomC{$\ptri{p}{t}{q}$}
				\AxiomC{$q \leq A(p)$}
				\RightLabel{(iterate)}
				\BinaryInfC{$\ptri{p}{t^*}{p + q}$}
			\end{prooftree}
		\end{minipage}
	\end{minipage}
	\end{framed}
	\caption{Proof system $\LCK_A$.}
	\label{fig:lcl}
\end{figure}

It turns out that the logic $\LCK_A$ is logically sound (we use ``logical'' soundness to avoid overloading the soundness 
of abstract semantics). 

\begin{restatable}[Logical Soundness of $\vdash^{{\scriptscriptstyle\mathrm{K}}}_{A}$]{theorem}{thmlsound} \label{thm:lsound}
	If $\ptri{p}{t}{q}$ then
	\begin{enumerate}[{\rm (i)}]
		\item \label{lsound_1} $q \leq_K \csema{t} p$;
		\item \label{lsound_2} $\asema{t} \alpha(p) = \alpha(q) = \alpha(\csema{t} p)$.
	\end{enumerate}
\end{restatable}
\begin{proof}
	The proof is adapted from \cite[Theorem 5.5]{bruni_logic_2021}. The first equality of (\ref{lsound_2}) is a consequence of the first one, (\ref{lsound_1}) and soundness:
	\[ \alpha(q) \leq \alpha(\csema{t} p) \leq \asema{t} \alpha(p) = \alpha(q) \]
	that implies $\alpha(q) = \alpha(\csema{t} p)$. For this reason we only need to prove $\asema{t} \alpha(p) = \alpha(q)$.
	
	\vspace{5px} \noindent The proof is on the structure of the derivation tree of $\ptri{p}{t}{q}$.
	
	\vspace{5px} \noindent (transfer)\,: (\ref{lsound_1}) is trivial, while (\ref{lsound_2}) is a consequence of the local completeness axiom $\mathbb{C}_p^A(\csema{a})$, $\asema{a} \alpha(p) = \alpha(\csema{a} \gamma(\alpha(p))) \stackrel{\text{LC}}{=} \alpha(\csema{a} p) = \alpha(q)$.
	
	\vspace{5px} \noindent (relax)\,: By induction $q' \leq \csema{t} p'$ so that $q \leq q' \leq \csema{t} p' \leq \csema{t} p$, implying (\ref{lsound_1}).
		By hypothesis we have $p' \leq p \leq A(p')$. By isotony of $\alpha(\cdot)$ and $\gamma(\cdot)$, and (\ref{eqn:uco_idem}) we have that $A(p') \leq A(p) \leq A(A(p')) = A(p')$, which means $A(p') = A(p)$. By the same reasoning $A(q) = A(q')$. By injectivity of $\gamma(\cdot)$ we have $\alpha(p') = \alpha(p)$ and ${\alpha(q') = \alpha(q)}$. This proves (\ref{lsound_2}).
		
	\vspace{5px} \noindent (seq)\,: (\ref{lsound_1}) holds by isotony of $\csema{\cdot}$ and by induction $q \leq \csema{t_2} r \leq \csema{t_2} \csema{t_1} p = \csema{t_1 \cdot t_2} p$, while (\ref{lsound_2}) holds by induction as ${\asema{t_1 \cdot t_2} \alpha(p) = \asema{t_2} \asema{t_1} \alpha(p) = \asema{t_2} \alpha(r) = \alpha(q)}$.
	
	\vspace{5px} \noindent (join)\,: The inductive hypothesis yields $q_1 \leq \csema{t_1} p$ and $q_2 \leq \csema{t_2} p$, thus by additivity of $\csema{\cdot}$, $q_1 + q_2 \leq \csema{t_1} p + \csema{t_2} p = \csema{t_1 + t_2} p$, implying (\ref{lsound_1}). (\ref{lsound_2}) instead can be obtained by induction and additivity of $\alpha(\cdot)$:
		\[ \asema{t_1 + t_2} \alpha(p) = \asema{t_1} \alpha(p) + \asema{t_2} \alpha(p) = \alpha(q_1) + \alpha(q_2) = \alpha(q_1 + q_2) \]

	\noindent (rec)\,: By induction we have that $r \leq \csema{t} p$. For any $a \in K$ it holds $a \leq a^*$ \cite[Section 2.1]{KOZEN1994366}, and by isotony of $\bdia{\cdot}$,
	\[ r \leq \csema{t} p = \bdia{\intrp{t}} p \leq \bdia{(\intrp{t})^*} p = \bdia{\intrp{t^*}} p = \csema{t^*} p \]
	Moreover, for any $a \in K$ it holds $a^* a^* = a^*$ \cite[Section 2.1]{KOZEN1994366}, so that
	\begin{align*}
		\csema{t^*} \csema{t^*} p &= \textrm{[By (\ref{eqn:csema_mul})]} \\		
		\csema{t^* \cdot t^*} p &= \\
		\bdia{\intrp{t^* \cdot t^*}} p &= \\
		\bdia{(\intrp{t})^* (\intrp{t})^*} p &= \\
		\csema{t^*} p
	\end{align*}
	By induction, it holds $q \leq \csema{t^*} (p+r) = \csema{t^*} p + \csema{t^*} r$, but ${\csema{t^*} r \leq \csema{t^*} \csema{t^*} p = \csema{t^*} p}$, thus implying (\ref{lsound_1}). (\ref{lsound_2}) can be retrieved by (\ref{lsound_1}), isotony of $\asema{\cdot}$, $\alpha(\cdot)$ and soundness
	\[ \alpha(q) \leq \alpha(\csema{t^*} p) \leq \asema{t^*} \alpha(p) \leq \asema{t^*} \alpha(p + r) = \alpha(q) \]
	
	\vspace{5px} \noindent (iterate)\,: For any $a \in K$ it holds $1 \leq a^*$ and $a \leq a^*$ \cite[Section 2.1]{KOZEN1994366}, thus by isotony of $\bdia{\cdot}$, $q \leq \csema{t} p = \bdia{\intrp{t}} p \leq \bdia{(\intrp{t})^*} p = \bdia{\intrp{t^*}} p \leq \csema{t^*} p$ and by (\ref{eqn:bdia_test}) and isotony of $\bdia{\cdot}$, $p = p \cdot 1 = \bdia{1} p \leq \bdia{(\intrp{t})^*} p = \csema{t^*} p$, thus implying (\ref{lsound_1}) $p + q \leq \csema{t^*} p$.
		To show (\ref{lsound_2}) we need the following preliminary fact:
		\begin{equation} \label{eqn:claim_7}
			\asema{t} \alpha(p) \leq \alpha(p) \Rightarrow \disj (\asema{t})^n \alpha(p) = \alpha(p)
		\end{equation}
		First, let us prove by induction on $n$ that $\alpha(p)$ is an upper-bound of the elements of the disjunction, i.e., $\forall n \in \mathbb{N} . \; (\asema{t})^n \alpha(p) \leq \alpha(p)$. The base case $n=0$ is true since $(\asema{t})^0 \alpha(p) = \alpha(p) \leq \alpha(p)$, while the inductive case holds by isotony of $\asema{\cdot}$, $(\asema{t})^{n+1} \alpha(p) = (\asema{t})^n \asema{t} \alpha(p) \leq (\asema{t})^n \alpha(p) \leq \alpha(p)$. Notice that $\alpha(p)$ is also the least upper bound because $(\asema{t})^0 \alpha(p) = \alpha(p)$. The hypothesis $q \leq A(p)$ implies $\alpha(q) \leq \alpha(p)$ by isotony of $\alpha(\cdot)$ and $\gamma(\cdot)$, (\ref{eqn:gamma_inject}) and (\ref{eqn:uco_idem}). Notice that ${\asema{t} \alpha(p) = \alpha(q) \leq \alpha(p)}$, meaning that we can apply (\ref{eqn:claim_7}):
		\begin{align*} 
		\asema{t^*} \alpha(p) = \disj (\asema{t})^n \alpha(p) = \alpha(p) = \alpha(p) + \alpha(q) = \alpha(p + q). 
		\tag*{\qed}
		\end{align*} 
\end{proof}

Analogously to what happens for $\LCL$, we can prove that $\LCK_A$ is logically complete under these two additional hypotheses: 
\begin{enumerate}[{\rm (A)}]
	\item \label{add1} The following infinitary rule is added to $\LCK_A$: 
	\begin{prooftree}
		\AxiomC{$\forall n \in \mathbb{N} . \; \ptri{p_n}{t}{p_{n+1}}$}
		\RightLabel{(limit)}
		\UnaryInfC{$\ptri{p_0}{t^*}{\disj p_n}$}
	\end{prooftree}
	Let us point out that the lub $\disj p_n$ always exists in $K$, as a consequence of 
	the CTC requirement on $K$.

	\item \label{add2} The concrete semantics of the primitive actions and tests occurring in the program are globally complete.
\end{enumerate}

It can be proved that the rule (limit) preserves logical soundness:
\begin{restatable}{lemma}{lmlsoundlimit} \label{lemma:lsound_limit}
	With the hypothesis of Theorem \ref{thm:lsound}, the rule (limit) preserves logical soundness.
\end{restatable}
\begin{proof}
	Let's assume that (limit) proves a triple $\tri{p_0}{t^*}{\disj p_n}$. The conditions (\ref{lsound_1}) is $\disj p_n \leq \csema{t^*} p_0$, which is equivalent to $\forall n \in \mathbb{N} . \; p_n \leq \csema{t^*} p_0$. The latter can be shown by induction on $n$. By (\ref{eqn:csema_star_unfold}), ${p_0 \leq p_0 + \csema{t} \csema{t^*} p_0 \leq \csema{t^*} p_0}$. Similarly for the inductive case $p_{n+1} \leq \csema{t} p_n \leq \csema{t} \csema{t^*} p_0 \leq p_0 + \csema{t} \csema{t^*} p_0 \leq \csema{t^*} p_0$, where the first inequality holds by structural induction, i.e., $\ptri{p_n}{t}{p_{n+1}}$ implies $\vtri{p_n}{t}{p_{n+1}}$.
	To show (\ref{lsound_2}) we need the auxiliary claim $\forall n \in \mathbb{N} . \; (\asema{t})^n \alpha(p_0) = \alpha(p_n)$. That can be proved by induction on $n$. The base case is trivially true as $(\asema{t})^0 \alpha(p_0) = \alpha(p_0)$. While the inductive case holds by structural induction, i.e., $\asema{t} \alpha(p_n) = \alpha(p_{n+1})$, as $(\asema{t})^{n+1} \alpha(p_0) = \asema{t} (\asema{t})^n \alpha(p_0) = \asema{t} \alpha(p_n) = \alpha(p_{n+1})$.
	Therefore by additivity of $\alpha(\cdot)$,
	\begin{align*}
 \asema{t^*} \alpha(p_0) = \disj (\asema{t})^n \alpha(p_0) = \disj \alpha(p_n) = \alpha(\disj p_n). \tag*{\qed}
 \end{align*}
\end{proof}

The condition \eqref{add2} entails global completeness:
\begin{restatable}{lemma}{lmglobcompl} \label{lemma:glob_compl}
	Let $K$ a backward-diamond CTC Kleene algebra, $A$ a Kleene abstract domain on $K$ and $T_{\Sigma,B}$ a KAT language on $K$. For any $\mathtt{t} \in T_{\Sigma, B}$ and $p \in \Test(K)$ we have
	\[ (\forall \mathtt{a} \in \Atom(\mathtt{t}) . \; \mathbb{C}^A(\csema{a})) \Rightarrow \asema{t} \alpha(p) = \alpha(\csema{t} p) \]
\end{restatable}
\begin{proof}
	The proof is on the structure of $\mathtt{t}$:
	
	\vspace{5px} \noindent ($\mathtt{a} \in \Atom$)\,:
	\begin{align*}
		\asema{a} \comp \alpha &= \\
		\alpha \comp \csema{a} \comp \gamma \comp \alpha &= \textrm{[By $\mathbb{C}^A(\csema{a})$, (\ref{eqn:gamma_inject})]} \\
		\alpha \comp \csema{a}
	\end{align*}		

	\noindent ($\mathtt{t_1 + t_2}$)\,: Let us denote with $p$ a generic test.
	\begin{align*}
		\asema{t_1 + t_2} \alpha(p) &= \\
		\asema{t_1} \alpha(p) + \asema{t_2} \alpha(p) &= \textrm{[By induction]} \\
		\alpha(\csema{t_1} p) + \alpha(\csema{t_2} p) &= \textrm{[By additivity of $\alpha(\cdot)$,$\csema{\cdot}$]} \\
		\alpha(\csema{t_1 + t_2} p)
	\end{align*}
	
	\noindent ($\mathtt{t_1 \cdot t_2}$)\,: Let us denote with $p$ a generic test.
	\begin{align*}
		\asema{t_1 \cdot t_2} \alpha(p) &= \\
		\asema{t_2} \asema{t_1} \alpha(p) &= \textrm{[By induction]} \\
		\asema{t_2} \alpha(\csema{t_1} p) &= \textrm{[By induction]} \\
		\alpha(\csema{t_2} \csema{t_1} p) &= \textrm{[By (\ref{eqn:csema_mul})]} \\
		\alpha(\csema{t_1 \cdot t_2} p)
	\end{align*}

	\noindent ($\mathtt{t_0^*}$)\,: The result is a consequence of the following claim
		\begin{equation} \label{eqn:claim_3}
			\forall n \in \mathbb{N} . \; (\asema{t_0})^n \alpha(p) = \alpha((\csema{t_0})^n p)
		\end{equation}
		that can be proved by induction on $n$. The base case is ${(\asema{t_0})^0 \alpha(p) = \alpha(p) = \alpha((\csema{t_0})^0 p)}$. While the inductive case can be proved as follows
		\[ (\asema{t_0})^{n+1} \alpha(p) = \asema{t_0} (\asema{t_0})^n \alpha(p) = \asema{t_0} \alpha((\csema{t_0})^n p) = \alpha((\csema{t_0})^{n+1} p) \]
		Therefore,
		\begin{align*}
			\asema{t_0^*} \alpha(p) &= \\
			\disj (\asema{t_0})^n \alpha(p) &= \textrm{[By (\ref{eqn:claim_3})]} \\
			\disj \alpha((\csema{t_0})^n p) &= \textrm{[By additivity of $\alpha(\cdot)$]} \\
			\alpha(\disj (\csema{t_0})^n p) &= \textrm{[By (\ref{eqn:csema_star})]}\\
			\alpha(\csema{t_0^*} p) \tag*{\qed}
		\end{align*} 
\end{proof}

\begin{restatable}[Logical Completeness of $\vdash^{{\scriptscriptstyle\mathrm{K}}}_{A}$]{theorem}{thmlcompl} \label{thm:lcompl}
	Assume that conditions \eqref{add1} and \eqref{add2} hold. If $\:\vtri{p}{t}{q}\:$ then $\:\ptri{p}{t}{q}$.
\end{restatable}
\begin{proof}
	The proof is organized in two parts. In the first part we prove that $\ptri{p}{t}{\csema{t} p}$. In the second one instead we prove
	\begin{prooftree}
		\AxiomC{$\ptri{p}{t}{\csema{t} p}$}
		\UnaryInfC{$\ptri{p}{t}{q}$}
	\end{prooftree}

	\noindent 1) The proof is by induction on the structure of $\mathtt{t} \in T_{\Sigma, B}$:
	
	\vspace{5px} \noindent ($\mathtt{a} \in \Atom$)\,: By the global completeness hypothesis we have $\mathbb{C}^A_p(\csema{a})$, hence (transfer) yields $\ptri{p}{a}{\csema{a}p}$.
	
	\vspace{5px} \noindent ($\mathtt{t_1 + t_2}$)\,: Let $q_i = \csema{t_i} p$, and by Lemma \ref{lemma:glob_compl} $\asema{t_i} \alpha(p) = \alpha(q_i)$, which means $\vtri{p}{t_i}{q_i}$, thus by induction $\ptri{p}{t_i}{q_i}$. Notice that by additivity of $\csema{\cdot}$, $q_1 + q_2 = \csema{t_1} p + \csema{t_2} p = \csema{t_1 + t_2} p$. Finally, (join) yields the result ${\ptri{p}{t_1 + t_2}{\csema{t_1 + t_2} p}}$
	
	\vspace{5px} \noindent ($\mathtt{t_1 \cdot t_2}$)\,: Let $q_1 = \csema{t_1} p$ and $q_2 = \csema{t_2} q_1$. By Lemma \ref{lemma:glob_compl} $\asema{t_1} \alpha(p) = \alpha(\csema{t_1} p)$ and $\asema{t_2} \alpha(q_1) = \alpha(\csema{t_2} q_1)$, which means $\vtri{p}{t_1}{q_1}$ and $\vtri{q_1}{t_2}{q_2}$, and by induction $\ptri{p}{t_1}{q_1}$ and $\ptri{q_1}{t_2}{q_2}$. Finally, (seq) yields the result $\ptri{p}{t_1 \cdot t_2}{\csema{t_1 \cdot t_2} p}$, where by (\ref{eqn:csema_mul}) $q_2 = \csema{t_2} \csema{t_1} p = \csema{t_1 \cdot t_2} p$.
	
	\vspace{5px} \noindent ($\mathtt{t_0^*}$)\,: Let $p_n = (\csema{t_0})^n p$. As preliminary step let us prove that $\vtri{p_n}{t_0}{p_{n+1}}$ for any $n$. (\ref{lsound_1}) holds because $p_{n+1} = \csema{t_0} p_n$, while (\ref{lsound_2}) is immediate by Lemma \ref{lemma:glob_compl}, $\asema{t_0} \alpha(p_n) = \alpha(\csema{t_0} p_n) = \alpha(p_{n+1})$. By induction we have $\forall n \in \mathbb{N} . \; \ptri{p_n}{t_0}{p_{n+1}}$, thus (limit) yields the result $\ptri{p}{t_0^*}{\csema{t_0^*} p}$ (notice that $p = (\csema{t_0})^0 p = p_0$ and by (\ref{eqn:csema_star}) $\csema{t_0^*} p = \disj (\csema{t_0})^n p = \disj p_n$).
	
	\vspace{5px} \noindent 2) By hypothesis $\vtri{p}{t}{q}$, thus $q \leq \csema{t} p$ and also $A(q) = A(\csema{t} p)$ so that we can use (relax) to obtain the result: $\ptri{p}{t}{q}$. \qed
\end{proof}

Summing up, this shows that the local completeness logic $\LCL$ introduced 
in \cite{bruni_logic_2021} can be made \emph{fully algebraic} by means of a natural interpretation 
on modal KATs with a backward diamond operator, still preserving its logical soundness and completeness, which are proved by
using just the algebraic axioms of this class of KATs. Hence, this shows that there is no need to leverage particular semantic 
properties of programs to reason on their local completeness.

\subsection{An Example of a Language-Theoretic KAT}\label{ltKAT-subsec}
To give an example digressing from programs and showing the generality of the KAT-based approach, we describe a language-theoretic model of Kleene algebra, early introduced by Kozen and Smith~\cite[Section 3]{goos_kleene_1997}.

Let $\Sigma = \{\mathtt{u}\}$ and $B = \{\mathtt{b}_{1}, \mathtt{b}_{2}\}$ be, resp., the sets of primitive actions and tests. An \emph{atom} is a string $c_1 c_2$, where $c_i \in \{\mathtt{b}_{i}, \overline{\mathtt{b}}_i\}$. If $c_i = \mathtt{b}_{i}$, where $i\in \{1,2\}$, then $\mathtt{b}_{i}$ appears positively in the atom $c_1 c_2$, while if $c_i = \mathtt{\overline{b}}_i$ it appears negatively. A \emph{guarded string} is 
either a single atom or 
a string $\alpha_0 \mathtt{a}_{1} \alpha_1 ... \mathtt{a}_{n} \alpha_n$, where $\alpha_i$ are atoms and $\mathtt{a}_i \in \Sigma$. If we are only interested in the first (last) atom of a guarded string $\alpha \mathtt{a}_{1} \alpha_1 ... \mathtt{a}_{n} \beta$, we may refer to it through the syntax $\alpha x$ ($x \beta$). 
Concatenation of guarded strings is given by a coalesced product operation $\diamond$, which is partially defined as follows:
\begin{equation*}
	x \alpha \diamond \beta y \triangleq
    \begin{cases*}
      x \alpha y & if $\alpha = \beta$ \\
      \textrm{undefined}        & otherwise
    \end{cases*}
\end{equation*}
The elements of this KAT $\mathcal{G}$ are sets of guarded strings. Thus, $+$ is set 
union, the product is defined as pointwise coalesced product:
\[ A \cdot B \ud \{ x \diamond y \mid x \in A, y \in B \}, \]
while the Kleene iteration is:
\( A^* \ud \bigcup_{n \in \mathbb{N}} A^n \).
The product identity  corresponds to the whole set of atoms 
$1_\mathcal{G}\ud\{\mathtt{b}_1 \mathtt{b}_2, \mathtt{b}_1 \overline{\mathtt{b}}_2, 
\overline{\mathtt{b}}_1 \mathtt{b}_2, \overline{\mathtt{b}}_1 \overline{\mathtt{b}}_2\}$, while
$0_\mathcal{G}$ is the empty set. The set of tests is $\Test(\mathcal{G})\ud \wp(1_\mathcal{G})$.

It turns out that $\mathcal{G}$ is a bdKAT, whose backward diamond is as follows: for all $a\in \cG$ and $p\in \Test(\cG)$, 
\begin{equation} \label{eqn:bdia_example}
	\bdia{a} p = \{ \beta \mid x \beta \in pa \}.
\end{equation}
\begin{proof}
	Let $r = \{ \beta \mid x \beta \in pa \}$. By Lemma~\ref{prop:bdia_char}, it is enough to show that $r$ is the least $q \in \Test(\mathcal{G})$ satisfying $pa = paq$.
		
	\noindent
	Let $x \beta \in pa$. By definition, $\beta \in r$ means that $x \beta \diamond \beta = x \beta$ is contained in $par$. This therefore means that $pa \leq par$. The opposite inequality is trivial since $r$ is a test, hence $r \leq 1_{\mathcal{G}}$, and by monotonicity of $\cdot$, we have that $pa \geq par$, thus implying $pa = par$.

	\noindent
	Assume now, by contradiction, that there exists $t \in \Test(\mathcal{G})$ such that $pa = pat$, $t \leq r$ and 
	$t\neq r$. This means that there is at least an atom $\beta$ in $r$ which is not in $t$. By definition of $r$, there is a guarded string $x \beta \in pa$. Since $pa = pat$, the last atom of all the guarded strings in $pa$ must be in $t$, but this does not hold for $x \beta$ as $\beta \notin t$. \qed
\end{proof}

\smallskip
Let us consider the evaluation function $G : \Atom \ra \mathcal{G}$ as defined in~\cite{goos_kleene_1997}:
\begin{align*}
	G(\mathtt{a}) &\triangleq \{ \alpha \mathtt{a} \beta \mid \alpha, \beta \in 1_\mathcal{G} \}, \\
	G(\mathtt{b}) &\triangleq \{ \alpha \in 1_\mathcal{G} \mid \mathtt{b} \textrm{ appears positively in } \alpha \}.
\end{align*}

We consider the abstract domain $A \ud \{ \top, e, o, \bot \}$ determined by the following abstraction 
$\alpha:\Test(\cG)\ra A$
and concretization $\gamma: A\ra \Test(\cG)$ maps:
\begin{align*}
	\alpha(p) \triangleq
	\begin{cases*}
		\bot & if $p = \emptyset$ \\
		e & if $\emptyset \subsetneq p \subseteq \{\mathtt{b}_{1} \mathtt{b}_2, \overline{\mathtt{b}}_1 \overline{\mathtt{b}}_2\}$ \\
		o & if $\emptyset \subsetneq p \subseteq \{\overline{\mathtt{b}}_1 \mathtt{b}_2, \mathtt{b}_1 \overline{\mathtt{b}}_2\}$ \\
		\top & otherwise
	\end{cases*}
	&&
	\gamma(p^\sharp) \triangleq
	\begin{cases*}
		\emptyset & if $p^\sharp = \bot$ \\
		 \{\mathtt{b}_{1} \mathtt{b}_2, \overline{\mathtt{b}}_1 \overline{\mathtt{b}}_2\} & if $p^\sharp = e$ \\
		 \{\overline{\mathtt{b}}_1 \mathtt{b}_2, \mathtt{b}_{1} \overline{\mathtt{b}}_2\} & if $p^\sharp = o$ \\
		 1_\mathcal{G} & if $p^\sharp = \top$
	\end{cases*}
\end{align*}
By counting, in an atom, the number of primitive tests that appear positively we obtain an integer that may be even or odd. Hence, 
this abstract domain $A$ represents the property of being even $e$ or odd $o$ of all the atoms occurring in a test $p\in \Test(\cG)$. 

By using our logic $\LCK$, we study the correctness of the program $\mathtt{r} \ud \mathtt{(u \cdot b_1)^*}\in \cG$, 
assuming a precondition $p \ud \{\mathtt{b}_1 \mathtt{b}_2, \overline{\mathtt{b}}_1\overline{\mathtt{b}}_2\}\in \Test(\cG)$ and a specification $\Spec \ud p = \gamma(e)$. Let us define two auxiliary tests: $q \ud \{\mathtt{b}_1 \mathtt{b}_2, \mathtt{b}_1 \overline{\mathtt{b}}_2 \}$, $s \ud \{ \mathtt{b}_1 \mathtt{b}_2, \mathtt{b}_1 \overline{\mathtt{b}}_2 , \overline{\mathtt{b}}_1\overline{\mathtt{b}}_2 \}$. Using the equation~\eqref{eqn:bdia_example}, we can easily check the following 
local completeness equations: 
\begin{align*}
	&\alpha(\csema{u} A(s)) = \alpha(\csema{u} 1_\mathcal{G}) = \alpha(1_\mathcal{G}) = \top = \alpha(1_\mathcal{G}) = \alpha(\csema{u} s) && \Rightarrow\; \mathbb{C}^A_s(\csema{u}) \\
	&\alpha(\csema{u} A(p)) = \alpha(\csema{u} p) = \alpha(1_\mathcal{G}) = \top = \alpha(1_\mathcal{G}) = \alpha(\csema{u} p) && \Rightarrow\; \mathbb{C}^A_p(\csema{u}) \\
	&\alpha(\csema{b_{\text{$1$}}} A(1_\mathcal{G})) = \alpha(\csema{b_{\text{$1$}}} 1_\mathcal{G}) = \alpha(q) = \top = \alpha(q) = \alpha(\csema{b_{\text{$1$}}} 1_\mathcal{G}) && \Rightarrow\; \mathbb{C}^A_{1_\mathcal{G}}(\csema{b_{\text{$1$}}})
\end{align*}
Therefore, we have  the following derivation in $\LCK_A$ of the triple $[p] \;\mathtt{r}\; [s]$:
\begin{scprooftree}{0.65}
		\AxiomC{$\mathbb{C}^A_p(\csema{u})$}		
		\RightLabel{(transfer)}
		\UnaryInfC{$\ptri{p}{u}{1_\mathcal{G}}$}
		\AxiomC{$\mathbb{C}^A_{1_\mathcal{G}}(\csema{b_{\text{$1$}}})$}		
		\RightLabel{(transfer)}
		\UnaryInfC{$\ptri{1_\mathcal{G}}{b_{\text{$1$}}}{q}$}
		\RightLabel{(seq)}
		\BinaryInfC{$\ptri{p}{u \cdot b_{\text{$1$}}}{q}$}
		\AxiomC{$\mathbb{C}^A_s(\csema{u})$}		
		\RightLabel{(transfer)}
		\UnaryInfC{$\ptri{s}{u}{1_\mathcal{G}}$}
		\AxiomC{$\mathbb{C}^A_{1_\mathcal{G}}(\csema{b_{\text{$1$}}})$}		
		\RightLabel{(transfer)}
		\UnaryInfC{$\ptri{1_\mathcal{G}}{b_{\text{$1$}}}{q}$}
		\RightLabel{(seq)}
		\BinaryInfC{$\ptri{s}{u \cdot b_{\text{$1$}}}{q}$}
		\AxiomC{$\!\!\!\!\!\!\!\!\! q \leq A(s)$}
		\RightLabel{(iterate)}
		\BinaryInfC{$\ptri{s}{(u \cdot b_{\text{$1$}})^*}{s}$}
		\RightLabel{(rec)}
		\BinaryInfC{$\ptri{p}{(u \cdot b_{\text{$1$}})^*}{s}$}
\end{scprooftree}

\medskip
Here, in accordance with the soundness Theorem~\ref{thm:lsound}, we have that $s \subseteq \csema{\mathtt{r}} p \subseteq A(s)$. Observe that $A(s) \nsubseteq \Spec$ holds, meaning that an abstract interpretation-based analysis fails to prove that 
the program $\mathtt{r}$ is correct for $\Spec$. However, unlike conventional abstract interpretation, $\LCK_A$ is capable to 
show that $s \smallsetminus \Spec = \{ \mathtt{b}_1 \overline{\mathtt{b}}_2 \}$ is indeed a true alert, meaning that the program $\mathtt{r}$ is really incorrect and the failure to prove its correctness was not due to a false alarm.

\subsection{Under-Approximation Logic}\label{sec:ual}
\begin{figure}[t]
	\centering
    \begin{framed}
    \vspace*{-5pt}
	\begin{minipage}{\textwidth}
		\begin{minipage}{\textwidth}
			\begin{prooftree}
				\AxiomC{$\mathtt{a} \in \Atom$}
				\RightLabel{(transfer)}
				\UnaryInfC{$\ptriul{p}{a}{\csema{a} p}$}
			\end{prooftree}
		\end{minipage}
	\end{minipage}
    \\[3pt]
	\begin{minipage}{\textwidth}
		\begin{minipage}{0.39\textwidth}
			\begin{prooftree}
			\AxiomC{\phantom{$\ptriul{p'}{t}{q'}$}} 
			\RightLabel{(empty)}
			\UnaryInfC{$\ptriul{p}{t}{0}$}
	         \end{prooftree}
		\end{minipage}
		\hfill
		\begin{minipage}{0.61\textwidth}
			\begin{prooftree}
				\AxiomC{$p' \leq p \qquad \ptriul{p'}{t}{q'} \qquad q \leq q'$}
				\RightLabel{(consequence)}
				\UnaryInfC{$\ptriul{p}{t}{q}$}
			\end{prooftree}
		\end{minipage}
	\end{minipage}
	\\[3pt]
	\begin{minipage}{\textwidth}
		\begin{minipage}{0.46\textwidth}
			\begin{prooftree}
				\AxiomC{$\ptriul{p_1}{t}{q_1}$}
				\AxiomC{$\ptriul{p_2}{t}{q_2}$}
				\RightLabel{(disj)}
				\BinaryInfC{$\ptriul{p_1 + p_2}{t}{q_1 + q_2}$}
			\end{prooftree}
		\end{minipage}
		\hfill
		\begin{minipage}{0.59\textwidth}
			\begin{prooftree}
				\AxiomC{$\ptriul{p}{t_{\text{$1$}}}{r}$}
				\AxiomC{$\ptriul{r}{t_{\text{$2$}}}{q}$}
				\RightLabel{(seq)}
				\BinaryInfC{$\ptriul{p}{t_{\text{$1$}} \cdot t_{\text{$2$}}}{q}$}
			\end{prooftree}
		\end{minipage}
	\end{minipage}
	\\[3pt]
	\begin{minipage}{\textwidth}
		\begin{minipage}{0.43\textwidth}
		\begin{prooftree}
			\AxiomC{\phantom{$\ptriul{p'}{t}{q'}$}} 
			\RightLabel{(iterate zero)}
			\UnaryInfC{$\ptriul{p}{t^*}{p}$}
	         \end{prooftree}
		\end{minipage}
		\hfill
		\begin{minipage}{0.59\textwidth}
			\begin{prooftree}
				\AxiomC{$\ptriul{p}{t^* \cdot t}{q}$}
				\RightLabel{(iterate non-zero)}
				\UnaryInfC{$\ptriul{p}{t^*}{q}$}
			\end{prooftree}
		\end{minipage}
	\end{minipage}
	\\[3pt]
	\begin{minipage}{\textwidth}
		\begin{minipage}{0.43\textwidth}
			\begin{prooftree}
				\AxiomC{$\forall n \in \mathbb{N} . \; \ptriul{p_n}{t}{p_{n+1}}$}
				\RightLabel{(back-v)}
				\UnaryInfC{$\ptriul{p_0}{t^*}{\disj p_n}$}
			\end{prooftree}
		\end{minipage}
		\hfill
		\begin{minipage}{0.59\textwidth}
			\begin{prooftree}
				\AxiomC{$\ptriul{p}{t_{\text{$i$}}}{q} \textrm{, with $i=1$ or $i=2$}$}
				\RightLabel{(choice)}
				\UnaryInfC{$\ptriul{p}{t_{\text{$1$}} + t_{\text{$2$}}}{q}$}
			\end{prooftree}
		\end{minipage}
	\end{minipage}
	\end{framed}
	\caption{Proof System $\UL$.}
	\label{fig:ul}
\end{figure}

O'Hearn~\cite{ohearn_incorrectness_2020} incorrectness logic ($\IL$) 
establishes two main novelties w.r.t.\ the seminal Hoare logic of program correctness~\cite{hoare_axiomatic_1969}:
	(1)~a valid  postcondition of an incorrectness triple for a program $P$  
	is an \emph{under-approxi\-mation} of the strongest postcondition of $P$, rather than an over-approximation of Hoare logic;
	(2)~incorrectness triples feature two postconditions: one corresponding to a ``normal'' program termination and 
	one corresponding to an erroneous termination.
Even if $\IL$ was originally defined with both those features, we first neglect the second one~---~i.e., we 
consider ``normal'' termination only~---~and we refer to the resulting program logic as Under-approximation Logic, denoted
by $\UL$. For the sake of clarity,  
Figure~\ref{fig:ul} recalls the $\UL$ proof system, adapted to our algebraic framework. We only focus on the 
``propositional'' fragment of this logic, meaning that the roles of all the special program commands (i.e., $\mathtt{error}$, $\mathtt{assume}$, $\mathtt{skip}$, $\mathtt{nondet}$ used in \cite{ohearn_incorrectness_2020}) and variable manipulations commands of incorrectness logic 
are played by some corresponding elements in $\Atom$. Hence, for all of them, 
the single rule (transfer) is unifying and enough for our purposes. 

Analogously to what has been proved in \cite[Section 6]{bruni_logic_2021} for $\LCL$, it turns out 
that the trivial abstraction, i.e., the abstract domain $A_{tr} \ud \{ \top \}$ that approximates  all the concrete elements to a single abstract element $\top$, allows us to show that  the instantiation $\LCK_{A_{tr}}$, with the additional rule (limit), boils down 
to $\UL$, namely, our $\LCK$ logic generalizes $\UL$, even when both are interpreted on KATs. 

\begin{restatable}[$\LCK_{A_{tr}} \equiv \UL$]{theorem}{thmlclulequiv} \label{thm:lcl_ul_equiv}
Let $K$ be a CTC bdKAT. Assume that $\LCK_{A_{tr}}$ includes the rule \mbox{\rm\text{(limit)}}.
	For any $p,q \in \Test(K)$ and $\mathtt{t} \in T_{\Sigma, B}$:
	\[ \ptritr{p}{t}{q} \quad \Leftrightarrow \quad \ptriul{p}{t}{q}. \]
\end{restatable}
\begin{proof}
	Notice that in the $\LCK_{A_{tr}}$ proof system, the premise $\mathbb{C}^{A_{tr}}_p(\csema{a})$ of (transfer) is always true because $A_{tr} \comp \csema{a} p = 1 = A_{tr} \comp \csema{a} \comp A_{tr}(p)$ and similarly the conditions $p \leq A_{tr}(p')$ and $q' \leq A_{tr}(q')$ of (relax) are always true because $A_{tr}(p) = 1$ which is the greatest element of $\Test(K)$.
	
	Moreover, there is a strong correspondence between the two proof systems: in particular the rules of Table \ref{tbl:rules_equiv} coincide.
	\begin{table}[t]
		\centering
		\begin{tabular}{|c|c|}
			\hline
			$\LCK_{A_{Tr}}$ & $\UL$ \\
			\hline \hline
			transfer & transfer \\
			relax & consequence \\
			seq & seq \\
			limit & back-v \\
			\hline 
		\end{tabular}
		\caption{Correspondence between the rules of $\LCK_{A_{tr}}$ and $\UL$.}
		\label{tbl:rules_equiv}
	\end{table}
	
	\noindent $\Rightarrow)$ As preliminary step, notice that derivation trees built by the proof of Theorem \ref{thm:lcompl} do not require the rules (rec) and (iterate), thus by soundness and completeness we can build a derivation tree of $\ptritr{p}{t}{q}$ which does not contain such rules.
	
	 The proof is by induction on the derivation tree of $\ptritr{p}{t}{q}$. The rules of Table \ref{tbl:rules_equiv} are immediate. The only rule left to prove is (join). By induction we have $\ptriul{p}{t_1}{q_1}$ and $\ptriul{p}{t_2}{q_2}$ and we need to show $\ptriul{p}{t_1 + t_2}{q_1 + q_2}$:
	\begin{prooftree}
		\AxiomC{$\ptriul{p}{t_1}{q_1}$}
		\RightLabel{(choice)}
		\UnaryInfC{$\ptriul{p}{t_1 + t_2}{q_1}$}
		\AxiomC{$\ptriul{p}{t_2}{q_2}$}
		\RightLabel{(choice)}
		\UnaryInfC{$\ptriul{p}{t_1 + t_2}{q_2}$}
		\RightLabel{(disj)}
		\BinaryInfC{$\ptriul{p}{t_1 + t_2}{q_1 + q_2}$}
	\end{prooftree}
	 
	\noindent $\Leftarrow)$ The proof is by induction on the derivation tree of $\ptriul{p}{t}{q}$. As in the previous case, the rules of Table \ref{tbl:rules_equiv} are immediate. The rules that are left to prove are (empty),(choice),(disj),(iterate zero) and (iterate non-zero).
	
	In order to carry out the derivation in those cases we exploit the completeness of $\LCK$. To do so we only need to verify condition (\ref{lsound_1}), because (\ref{lsound_2}) is trivially true: $\asema{t} \alpha(p) = \top = \alpha(q) = \top = \alpha(\csema{t} p)$ for any $p,q \in \Test(K)$ and $\mathtt{t} \in T_{\Sigma, B}$.

	\vspace{5px} \noindent (empty)\,: $\vtritr{p}{t}{0}$ is trivial since $0 \leq \csema{t} p$ is always true, therefore by completeness of $\LCK_{A_{tr}}$, $\ptritr{p}{t}{0}$.
	
	\vspace{5px} \noindent (choice)\,: By induction we have, without loss of generality, $\ptritr{p}{t_1}{q}$ and we need to show $\ptritr{p}{t_1 + t_2}{q}$. Notice that for any $\mathtt{t_2} \in T_{\Sigma, B}$, $\vtritr{p}{t_2}{0}$ because (\ref{lsound_1}) is trivially true: $0 \leq \csema{t_2} p$, therefore by completeness of $\LCK_{A_{tr}}$ $\ptritr{p}{t_2}{0}$. Finally, (join) yields the result $\ptritr{p}{t_1 + t_2}{q}$.
		
	\vspace{5px} \noindent (disj)\,: By induction we have $\ptritr{p_1}{t}{q_1}$ and $\ptritr{p_2}{t}{q_2}$ and we need to show $\ptritr{p_1 + p_2}{t}{q_1 + q_2}$. By soundness we have $\vtritr{p_1}{t}{q_1}$ and $\vtritr{p_2}{t}{q_2}$, meaning $q_1 \leq \csema{t} p_1$ and $q_2 \leq \csema{t} p_2$, so that $q_1 + q_2 \leq \csema{t} (p_1 + p_2)$. Finally, by completeness we retrieve the result: $\ptritr{p_1 + p_2}{t}{q_1 + q_2}$.
	
	\vspace{5px} \noindent (iterate zero)\,: By (\ref{eqn:csema_star_unfold}) it holds $\csema{t^*} p \geq p$ which means $\vtritr{p}{t^*}{p}$, and by completeness we obtain the result $\ptritr{p}{t^*}{p}$.
	
	\vspace{5px} \noindent (iterate non-zero)\,: By (\ref{eqn:csema_star_unfold}) it holds $\csema{t^*} p \geq \csema{t^* \cdot t} p$. By hypothesis $\ptritr{p}{t^* \cdot t}{q}$, hence by soundness $q \leq \csema{t^* \cdot t} p$. Combining the two we yields $q \leq \csema{t^*} p$, which means $\vtritr{p}{t^*}{q}$ and by completeness $\ptritr{p}{t^*}{q}$. \qed
\end{proof}

Moreover, since the abstraction map defining $A_{tr}$ is $\alpha_{A_{tr}}= \lambda x.\top$, 
we have that condition~\eqref{val-cond2} of  Definition~\ref{def:validity} for the validity of a  $\LCK$ triple 
trivially holds, that is, $\asema{t} \alpha(p) = \top = \alpha(q) = \alpha(\csema{t}^K p)$. 
This therefore entails that 
\begin{equation}\label{equiv3}
\vtritr{p}{t}{q}\quad\Leftrightarrow\quad \vtriul{p}{t}{q}.
\end{equation}

This allows us to retrieve the soundness and completeness results of incorrectness logic \cite{ohearn_incorrectness_2020} as a consequence of those for $\LCK$.

\begin{restatable}{corollary}{corulsndcompl} \label{cor:ul_snd_compl}
	Under the same hypotheses of Theorem \ref{thm:lcl_ul_equiv}, the proof system $\UL$ is sound and complete, that is,
	\(\ptriul{p}{t}{q} \; \Leftrightarrow \;\: \vtriul{p}{t}{q} \).
\end{restatable}
\begin{proof}
	\begin{align*}
		\ptriul{p}{t}{q} &\Leftrightarrow \qquad \textrm{[by Theorem \ref{thm:lcl_ul_equiv}]} \\
		\ptritr{p}{t}{q} &\Leftrightarrow \qquad \textrm{[by Theorems \ref{thm:lsound} and \ref{thm:lcompl}]} \\
		\vtritr{p}{t}{q} &\Leftrightarrow \qquad \textrm{[by \eqref{equiv3}]}\\
		\vtriul{p}{t}{q} & \tag*{\qed}
	\end{align*}
\end{proof}

\section{Incorrectness Logic in KAT} \label{sect:inc}

Incorrectness logic $\IL$ has been introduced by O'Hearn \cite{ohearn_incorrectness_2020} as a natural under-approxi\-mating 
counterpart of the pivotal Hoare correctness logic \cite{hoare_axiomatic_1969}, and quickly attracted a lot of research interest \cite{RaadBDO22,RaadBDDOV20,bugs2022,yan2022,Poskitt21}.
Incorrectness logic distinguishes two postconditions corresponding to normal and erroneous/abnormal program termination. Here, we generalize the algebraic formulation of our $\LCK$ logic to support abnormal termination.
We follow the approach of M{\"{o}}ller, O'Hearn and Hoare~\cite{fahrenberg_algebra_2021}, namely, each language term is interpreted as a pair of KAT elements which model the normal and abnormal execution. The evaluation function has type $u : \Atom \rightarrow (K \times K)$, while the interpretation function has type $\intrp{\cdot} : T_{\Sigma, B} \rightarrow (K \times K)$. As a shorthand $\intrp{\cdot}$ can be subscripted with \ok or \err to denote, resp., its first normal and second erroneous component. The definition is as follows:
\begin{equation}\label{def:il_intrp}
\begin{aligned}
	\intrp{a} &\triangleq u(\mathtt{a}) 
	\\
	\intrp{t_{\text{$1$}} + t_{\text{$2$}}} &\triangleq (\intrpok{t_{\text{$1$}}} + \intrpok{t_{\text{$2$}}}, \intrperr{t_{\text{$1$}}} + \intrperr{t_{\text{$2$}}}) \\
	\intrp{t_{\text{$1$}} \cdot t_{\text{$2$}}} &\triangleq (\intrpok{t_{\text{$1$}}} \cdot \intrpok{t_{\text{$2$}}}, \intrperr{t_{\text{$1$}}} + \intrpok{t_{\text{$1$}}} \cdot \intrperr{t_{\text{$2$}}}) \\
	\intrp{t^*} &\triangleq (\green{\llbracket \mathtt{t} \rrbracket_{u_{\ok}}^*}, \green{\llbracket \mathtt{t} \rrbracket_{u_{\ok}}^*} \cdot \intrperr{t})
\end{aligned}
\end{equation}
Following the original definition of $\IL$, the precondition encodes an \ok condition only, while the postcondition contains both an \ok and an \err component. Hence, the latter is given by a pair $(p,q)\in\Test(K) \times \Test(K)$, typically denoted by $\ilc{p}{q}$. The concrete semantics \( \csema{\cdot} : T_{\Sigma, B} \rightarrow (\Test(K) \rightarrow (\Test(K) \times \Test(K))) \) is
defined as
\[ \csema{t} p \triangleq \ilc{\bdia{\intrpok{t}}p}{\bdia{\intrperr{t}}p} \]
To refer to one of its components, $\csema{\cdot}$ can be subscripted with \ok or \err, e.g., $\csemaok{t} p$.

Like for the $\LCK$ formulation, the concrete semantic of $\LCIL$ leverages the backward-diamond operator, therefore it inherits its properties (e.g., isotony, additivity). Moreover, by (\ref{eqn:bdia_add_cmd}), (\ref{eqn:bdia2}) we have that
\begin{align}
	\label{eqn:csema_err_mul} \csemaerr{t_1 \cdot t_2} p &= \csemaerr{t_1} p + \csemaerr{t_2} \csemaok{t_1} p \\
	\label{eqn:csema_err_star} \csemaerr{t^*} p &= \csemaerr{t} \csemaok{t^*} p
\end{align}

Given a Kleene abstract domain $A$ on $K$, the corresponding abstract semantics 
\( \asema{\cdot} : T_{\Sigma, B} \rightarrow (A \rightarrow (A \times A)) \) is
defined as follows:
\begin{equation} \label{def:abs_sem_il}
\scalebox{0.99}
{$
\begin{aligned}
	\asema{a} p^\sharp &\triangleq \ilc{\alpha(\csemaok{a} \gamma(p^\sharp))}{\alpha(\csemaerr{a} \gamma(p^\sharp))} 
	\\
	\asema{t_{\text{$1$}} + t_{\text{$2$}}} p^\sharp &\triangleq \ilc{\asemaok{t_{\text{$1$}}} p^\sharp + \asemaok{t_{\text{$2$}}} p^\sharp}{\asemaerr{t_{\text{$1$}}} p^\sharp + \asemaerr{t_{\text{$2$}}} p^\sharp} \\
	\asema{t_{\text{$1$}} \cdot t_{\text{$2$}}} p^\sharp &\triangleq \ilc{\asemaok{t_{\text{$2$}}}(\asemaok{t_{\text{$1$}}} p^\sharp)}{\asemaerr{t_{\text{$1$}}} p^\sharp + \asemaerr{t_{\text{$2$}}}(\asemaok{t_{\text{$1$}}} p^\sharp)} \\
	\asema{t^*} p^\sharp &\triangleq \ilc{{\textstyle\disj} (\asemaok{t})^n p^\sharp}{\asemaerr{t} {\textstyle\disj}  (\asemaok{t})^n p^\sharp}
\end{aligned}
$}
\end{equation}

The \ok part coincides with the semantics of $\LCK$, while the \err component puts in place some differences. In particular, 
the composition exhibits a short-circuiting behavior, meaning that an error in the first command aborts the execution without executing the second one, while the Kleene star allows an error to occur after some error-free iterations. It is straightforward to check that this definition of abstract semantics is monotonic and sound.

The proof system $\LCK$ can be extended with incorrectness triples. In particular, a triple $\triokerr{p}{t}{q}{r}$ is valid 
if the standard validity conditions hold for both \ok and \err.
\begin{definitionrm}[Incorrectness Triple]
Let $K$ be a CTC bdKAT $K$ and $T_{\Sigma,B}$ be a language interpreted on $K$. An 
\emph{incorrectness triple} is either $\triok{p}{t}{q}$ or $\trierr{p}{t}{r}$, where $p,q,r \in \Test(K)$ and $\mathtt{t} \in T_{\Sigma,B}$.

\noindent
Let $A$ be a Kleene abstract domain on $K$ with abstraction map $\alpha:\Test(K)\ra A$.   

\begin{itemize}
\item The triple $\triok{p}{t}{q}$ is valid if: (1)~$q \leq \csemaok{t} p$, and (2)~$\asemaok{t} \alpha(p) = \alpha(q) = \alpha(\csemaok{t} p)$.
\item The triple $\trierr{p}{t}{r}$ is valid if: (1)~$r \leq \csemaerr{t} p$, and 
(2)~$\asemaerr{t} \alpha(p) = \alpha(r) = \alpha(\csemaerr{t} p)$.
\item
A triple $\triokerr{p}{t}{q}{r}$ is valid when both $\triok{p}{t}{q}$ and $\trierr{p}{t}{r}$ are valid. 
In particular, if $q = r$ then the triple $\triboth{p}{t}{q}$ is valid. \qed
\end{itemize}
\end{definitionrm}

The proof system $\LCIL_A$ defining the local completeness  incorrectness logic is given in
Figure \ref{fig:ilcl}. 

\begin{figure}[t]
	\centering
	\begin{framed}
	    \vspace*{-5pt}
	\begin{minipage}{\textwidth}
		\begin{minipage}{\textwidth}
			\begin{prooftree}
				\AxiomC{$\mathtt{a} \in \Atom \qquad \mathbb{C}^A_p(\csemaok{a}) \qquad \mathbb{C}^A_p(\csemaerr{a})$}
				\RightLabel{(transfer)}
				\UnaryInfC{$\ptriokerr{p}{a}{\csemaok{a} p}{\csemaerr{a} p}$}
			\end{prooftree}
		\end{minipage}
		\hfill
	\end{minipage}
	\\[3pt]
	\begin{minipage}{\textwidth}
		\begin{prooftree}
			\AxiomC{$ p' \leq p \leq A(p') \qquad \ptriboth{p'}{t}{q'} \qquad q \leq q' \leq A(q)$}
			\RightLabel{(relax)}
			\UnaryInfC{$\ptriboth{p}{t}{q}$}
		\end{prooftree}			
	\end{minipage}
	\\[3pt]
	\begin{minipage}{\textwidth}
		\begin{minipage}{\textwidth}
			\begin{prooftree}
				\AxiomC{$\ptriok{p}{t_{\text{$1$}}}{r}$}
				\AxiomC{$\ptriok{r}{t_{\text{$2$}}}{q}$}
				\RightLabel{(seq-ok)}
				\BinaryInfC{$\ptriok{p}{t_{\text{$1$}} \cdot t_{\text{$2$}}}{q}$}
			\end{prooftree}	
		\end{minipage}
	\\[3pt]
		\begin{minipage}{\textwidth}
			\begin{prooftree}
				\AxiomC{$\ptriokerr{p}{t_{\text{$1$}}}{q}{r}$}
				\AxiomC{$\ptrierr{q}{t_{\text{$2$}}}{s}$}
				\RightLabel{(seq-err)}
				\BinaryInfC{$\ptrierr{p}{t_{\text{$1$}} \cdot t_{\text{$2$}}}{r + s}$}
			\end{prooftree}			
		\end{minipage}
	\end{minipage}
	\\[3pt]
	\begin{minipage}{\textwidth}
		\begin{minipage}{\textwidth}
			\begin{prooftree}
				\AxiomC{$\ptriok{p}{t^*}{q}$}
				\AxiomC{$\ptrierr{q}{t}{r}$}
				\RightLabel{(rec-err)}
				\BinaryInfC{$\ptrierr{p}{t^*}{r}$}
			\end{prooftree}	
		\end{minipage}
		\end{minipage}
	\\[3pt]
		\begin{minipage}{\textwidth}
		\begin{minipage}{\textwidth}
			\begin{prooftree}
				\AxiomC{$\ptriboth{p}{t_{\text{$1$}}}{q_1}$}
				\AxiomC{$\ptriboth{p}{t_{\text{$2$}}}{q_2}$}
				\RightLabel{(join)}
				\BinaryInfC{$\ptriboth{p}{t_{\text{$1$}} + t_{\text{$2$}}}{q_1 + q_2}$}
			\end{prooftree}			
		\end{minipage}
		\\[3pt]
	\begin{minipage}{\textwidth}
		\begin{prooftree}
			\AxiomC{$\forall n \in \mathbb{N} . \; \ptriok{p_n}{t}{p_{n+1}}$}
			\RightLabel{(limit)}
			\UnaryInfC{$\ptriok{p_0}{t^*}{\disj p_n}$}
		\end{prooftree}	
	\end{minipage}
		\end{minipage}
		\end{framed}
	\caption{Proof system $\LCIL_A$.}
	\label{fig:ilcl}
\end{figure}

\begin{restatable}[Logical Soundness of $\LCIL_A$]{theorem}{thmlsoundil} \label{thm:lsound_il}
	The triples provable in $\LCIL_A$ are valid.
\end{restatable}
\begin{proof}
	Definitions of $\intrpok{\cdot}$ and $\asemaok{\cdot}$ coincide, resp., with $\intrp{\cdot}$ and $\asema{\cdot}$ of $\LCK$, and the \ok part of $\LCIL$ is a subset of $\LCK$\footnote{Here, we refer to the proof system $\LCK$ with the additional rule (limit).}. For this reason the soundness of the \ok part of the triple follows by Theorem \ref{thm:lsound}. Moreover, the base case $\mathtt{t} \equiv \mathtt{a} \in \Atom$ and the inductive case $\mathtt{t} \equiv \mathtt{t_1 + t_2}$ of $\intrperr{\cdot}$ and $\asemaerr{\cdot}$ are completely symmetric to their \ok counterparts. For this reason the \err part of (transfer),(join) can be shown in the same way as the \ok one. Likewise, the soundness of the \err part of (relax) follows as for the \ok one.

The only cases that are left to prove are (seq-err) and (rec-err). As in Theorem \ref{thm:lsound}, notice that the first equality of (2) entails the second one, i.e., if $\asemaerr{t} \alpha(p) = \alpha(r)$ is true, then $\asemaerr{t} \alpha(p) = \alpha(\csemaerr{t} p)$. For this reason it is enough to show only the former.

	\vspace{5px} \noindent (seq-err)\,: By induction we have $q \leq \csemaok{t_1} p$, $r \leq \csemaerr{t_1} p$ and $s \leq \csemaerr{t_2} q$, thus $s \leq \csemaerr{t_2} q \leq \csemaerr{t_2} \csemaok{t_1} p$. By (\ref{eqn:csema_err_mul}) we have that $\csemaerr{t_1 \cdot t_2} p = \csemaerr{t_1} p + \csemaerr{t_2} \csemaok{t_1} p \geq r + s$, thus implying (1).
	By induction we have $\asemaok{t_1} \alpha(p) = \alpha(q)$, $\asemaerr{t_1} \alpha(p) = \alpha(r)$ and $\asemaerr{t_2} \alpha(q) = \alpha(s)$, so that $\asemaerr{t_2} \asemaok{t_1} \alpha(p) = \asemaerr{t_2} \alpha(q) = \alpha(s)$. By definition of ${\asemaerr{t_1 \cdot t_2}}$, we can prove (2) as follows
	\begin{align*}
		\asemaerr{t_1 \cdot t_2} \alpha(p) &= \\
		\asemaerr{t_1} \alpha(p) + \asemaerr{t_2} \asemaok{t_1} \alpha(p) &= \\
		\alpha(r) + \alpha(s) &= \textrm{[Additivity of $\alpha(\cdot)$]} \\
		\alpha(r + s).
	\end{align*}
	
	\vspace{5px} \noindent (rec-err)\,: By induction we have $q \leq \csemaok{t^*} p$, $r \leq \csemaerr{t} q$. By (\ref{eqn:csema_err_star}) we have that ${\csemaerr{t^*} p = \csemaerr{t} \csemaok{t^*} p \geq \csemaerr{t} q \geq r}$, thus implying (1). By induction we have that $\asemaok{t^*} \alpha(p) = \alpha(q)$ and $\asemaerr{t} \alpha(q) = \alpha(r)$. Using the definition of $\asemaerr{t^*}$ we can prove (2) as $\asemaerr{t^*} \alpha(p) = \asemaerr{t} \asemaok{t^*} \alpha(p) = \asemaerr{t} \alpha(q) = \alpha(r)$. \qed
\end{proof}

Furthermore, it turns out that $\LCIL_A$ is logically complete. 

\begin{restatable}{lemma}{lemmaglobcomplil}\label{lemma:glob_compl_il}
	Let $A$ be a Kleene abstract domain on a CTC bdKAT $K$ and $T_{\Sigma,B}$ be a language interpreted on $K$. If the atoms in $\mathtt{t} \in T_{\Sigma, B}$ are globally complete, i.e., for all $\mathtt{a} \in \Atom(\mathtt{t})$, 
	$\mathbb{C}^A(\csemaok{a})$ and  $\mathbb{C}^A(\csemaerr{a}) \),
	then, for any $p \in \Test(K)$,
	$\asemaok{t} \alpha(p) = \alpha(\csemaok{t} p)$ and ${\asemaerr{t} \alpha(p) = \alpha(\csemaerr{t} p)}$.
\end{restatable}
\begin{proof}
Analogous to the proof of Lemma \ref{lemma:glob_compl}.
\qed
\end{proof}

\begin{restatable}[Logical Completeness of $\LCIL_A$]{theorem}{thmlcomplil} \label{thm:lcompl_il}
Let $A$ be a Kleene abstract domain on a CTC bdKAT $K$ and $T_{\Sigma,B}$ be a language interpreted on $K$. Assume that 
the atoms in $\mathtt{t} \in T_{\Sigma, B}$ are globally complete, i.e., for all $\mathtt{a} \in \Atom(\mathtt{t})$, 
	{\rm $\mathbb{C}^A(\csemaok{a})$} and  {\rm $\mathbb{C}^A(\csemaerr{a})$} hold.
	If $\:\triokerr{p}{t}{q}{r}\:$ is valid, then it is provable in $\LCIL_A$.
\end{restatable}
\begin{proof}
	As discussed in the proof of Theorem \ref{thm:lsound_il}, the \ok part of the semantic coincides with the $\LCK$ logic, so that $\triok{p}{t}{q}$ is provable by the completeness theorem of $\LCK_A$. For this reason only $\trierr{p}{t}{r}$ must be shown. The proof is by induction on the structure of $\mathtt{t} \in T_{\Sigma, B}$. By construction, the \err base case $\mathtt{t} \equiv \mathtt{a} \in \Atom$ and the inductive case $\mathtt{t} \equiv \mathtt{t_1 + t_2}$ are symmetric to their \ok counterparts, so that the proof of those cases is equivalent to Theorem \ref{thm:lcompl}. Therefore we only need to prove $\cdot$ and $^*$.
	
	\vspace{5px} \noindent ($\mathtt{t_1 \cdot t_2}$)\,: Consider the following assignments (the equality after the $\Rightarrow$ sign is true by Lemma \ref{lemma:glob_compl_il}):
	\begin{align*}
		s = \csemaok{t_1} p &\Rightarrow \asemaok{t_1} \alpha(p) = \alpha(\csemaok{t_1} p) = \alpha(s) \\
		t = \csemaerr{t_1} p &\Rightarrow \asemaerr{t_1} \alpha(p) = \alpha(\csemaerr{t_1} p) = \alpha(t) \\
		u = \csemaerr{t_2} s &\Rightarrow \asemaerr{t_2} \alpha(s) = \alpha(\csemaerr{t_2} s) = \alpha(u)
	\end{align*}
	and by induction we have that
	\begin{align*}
		\vtriok{p}{t_1}{s} &\Rightarrow \; \ptriok{p}{t_1}{s} \\
		\vtrierr{p}{t_1}{t} &\Rightarrow \; \ptrierr{p}{t_1}{t} \\
		\vtrierr{s}{t_2}{u} &\Rightarrow \; \ptrierr{s}{t_2}{u}
	\end{align*}
	Then by (seq-err) the triple $\trierr{p}{t_1 \cdot t_2}{t + u}$ is provable. By hypothesis we have
	\[ r \leq \csemaerr{t_1 \cdot t_2} p \stackrel{(\ref{eqn:csema_err_mul})}{=} \csemaerr{t_1} p + \csemaerr{t_2} \csemaok{t_1} p = t + u \]
	and
	\[ \alpha(r) = \alpha(\csemaerr{t_1 \cdot t_2} p) = \alpha(t + u) \]
	so that we can apply (relax) to obtain the result $\ptrierr{p}{t_1 \cdot t_2}{r}$.
	
	\vspace{5px} \noindent ($\mathtt{t_0^*}$)\,: Let $q = \csemaok{t_0^*} p$ and $s = \csemaerr{t_0} \csemaok{t_0^*} p$. By Lemma \ref{lemma:glob_compl_il} $\asemaok{t_0^*} \alpha(p) = \alpha(q)$ so that $\vtriok{p}{t_0^*}{q}$ and the completeness of the \ok part yields $\ptriok{p}{t_0^*}{q}$. By Lemma \ref{lemma:glob_compl_il} $\asemaerr{t_0} \alpha(q) = \alpha(s)$, meaning that $\vtrierr{q}{t_0}{s}$ and by induction $\ptrierr{q}{t_0}{s}$. By (rec-err) we obtain $\ptrierr{p}{t_0^*}{s}$.
	
	\noindent
	In order to use (relax), we need to ensure that the additional conditions $r \leq s$ and $\alpha(r) = \alpha(s)$ are satisfied. By $\vtrierr{p}{t_0^*}{r}$ and (\ref{eqn:csema_err_star}), $r \leq \csemaerr{t_0^*} p = s$ and $\alpha(r) = \alpha(\csemaerr{t_0^*} p) = \alpha(\csemaerr{t_0} \csemaok{t_0^*} p) = \alpha(s)$, so that (relax) yields the result $\ptrierr{p}{t_0^*}{r}$. \qed
\end{proof}

\begin{examplebf}
	Consider a relational bdKAT $K\ud \wp(\bZ\times \bZ)$ 
	on the set of integers $\bZ$, where $\textbf{1}_K\ud\{\tuple{z,z} \mid z\in \bZ\}$ and $\textbf{0}_K\ud\varnothing$, 
	and the  standard integer interval abstraction $\Int$~\cite{CC77,CC79}. Let us consider a language with primitive actions 
	$\Sigma \ud \{\mathtt{x:=x+1}, \Error\}$. The evaluation function $u:\Sigma \cup B \rightarrow K_\ok\times K_\err$ is defined as expected: 
	\begin{align*}
		u(\mathtt{x:=x+1}) &= (\{\tuple{z,z+1} \mid z \in \mathbb{Z}\}, \textbf{0}_K), && u(\Error) = (\textbf{0}_K, \textbf{1}_K).
	\end{align*}
	We study the correctness of the program ${\mathtt{r} \equiv ((\mathtt{x:=x+1}) +
	\Error)^*}$, 
	for the precondition $p \triangleq \{\tuple{0,0},\tuple{2,2}\}$ and the specification $\Spec \triangleq (\ok\! : \{\tuple{z,z} \mid z \geq 0 \}, \err\! :\textbf{0}_K)$. Let us 
	define an auxiliary sequence of tests $p_n \ud \{\tuple{n,n},\tuple{n+2,n+2}\}$ and
	$s \ud \{ \tuple{z,z} \mid z \geq 0 \}$.
	
	\noindent
	We can easily check the local completeness of the atoms by exploiting the
	characterization of the backward diamond operator in relational
	KATs of Lemma \ref{lemma:bdia_rela}.

	\noindent	
	We therefore have the following derivation in $\LCIL_{\Int}$ for $\mathtt{r}$:
	
	\begin{scprooftree}{0.75}
		\AxiomC{$\mathbb{C}^{\Int}_{p_n}(\csemaok{x := x+1})$}
		\AxiomC{$\mathbb{C}^{\Int}_{p_n}(\csemaerr{x := x+1})$}
		\RightLabel{(transfer)}
		\BinaryInfC{$\patriok{p_n}{x:=x+1}{p_{n+1}}{\Int}$}
		\AxiomC{$\mathbb{C}^{\Int}_{p_n}(\csemaok{\Error})$}
		\AxiomC{$\mathbb{C}^{\Int}_{p_n}(\csemaerr{\Error})$}
		\RightLabel{(transfer)}
		\BinaryInfC{$\patriok{p_n}{\Error}{0}{\Int}$}
		\RightLabel{(choice)}
		\BinaryInfC{$\patriok{p_n}{(\mathtt{x:=x+1}) + \Error}{p_{n+1}}{\Int}$}
		\RightLabel{(limit)}
		\UnaryInfC{$\dagger$}
	\end{scprooftree}
	\begin{scprooftree}{0.75}
		\AxiomC{$\mathbb{C}^{\Int}_{s}(\csemaok{x := x+1})$}
		\AxiomC{$\mathbb{C}^{\Int}_{s}(\csemaerr{x := x+1})$}
		\RightLabel{(transfer)}
		\BinaryInfC{$\patrierr{s}{x:=x+1}{0}{\Int}$}
		\AxiomC{$\mathbb{C}^{\Int}_{s}(\csemaok{\Error})$}
		\AxiomC{$\mathbb{C}^{\Int}_{s}(\csemaerr{\Error})$}
		\RightLabel{(transfer)}
		\BinaryInfC{$\patrierr{s}{\Error}{s}{\Int}$}
		\RightLabel{(choice)}
		\BinaryInfC{$\ddagger$}
	\end{scprooftree}	
	\begin{scprooftree}{0.75}
		\AxiomC{$\dagger$}
		\RightLabel{(limit)}
		\UnaryInfC{$\patriok{p_0}{((\mathtt{x:=x+1}) + \Error)^*}{s}{\Int}$}
		\RightLabel{(choice)}
		\AxiomC{$\ddagger$}
		\UnaryInfC{$\patrierr{s}{(\mathtt{x:=x+1}) + \Error}{s}{\Int}$}
		\RightLabel{(rec-err)}
		\BinaryInfC{$\patrierr{p_0}{((\mathtt{x:=x+1}) + \Error)^*}{s}{\Int}$}
	\end{scprooftree}

\noindent
	By soundness of $\LCIL_{\Int}$ in Theorem~\ref{thm:lsound_il}, the program $\mathtt{r}$ satisfies the \ok part of $\Spec$ 
	because $$\csemaok{r} p \subseteq \Int(s) = s \subseteq s = \Spec_{\ok}.$$ However, the \err part is not satisfied as $\Int(s) = s \nsubseteq \varnothing = \textbf{0}_K = \Spec_{\err}$. Moreover, $\LCIL_{\Int}$ also catches true alerts as $s \smallsetminus \Spec_{\err} = s$.
	\qed
\end{examplebf}

\subsection{Relationship with Incorrectness logic} \label{sect:lcil_il}

\begin{figure}[t]
	\centering
	\begin{framed}
	    \vspace*{-5pt}
	\begin{minipage}{\textwidth}
		\begin{minipage}{0.5\textwidth}
			\begin{prooftree}
				\AxiomC{$\mathtt{a} \in \Atom$}
				\RightLabel{(transfer)}
				\UnaryInfC{$\ptriokerril{p}{a}{\csemaok{a} p}{\csemaerr{a} p}$}
			\end{prooftree}
		\end{minipage}
		\hfill
		\begin{minipage}{0.35\textwidth}
		\begin{prooftree}
		\AxiomC{\phantom{$\mathtt{a} \in \Atom$}}
				\RightLabel{(empty)}
				\UnaryInfC{$\ptribothil{p}{t}{0}$}
		\end{prooftree}
		\end{minipage}
	\end{minipage}
	\\[3pt]
	\begin{minipage}{\textwidth}
		\begin{prooftree}
			\AxiomC{$p' \leq p \qquad \ptribothil{p'}{t}{q'} \qquad q \leq q'$}
			\RightLabel{(consequence)}
			\UnaryInfC{$\ptribothil{p}{t}{q}$}
		\end{prooftree}
	\end{minipage}
	\\[3pt]
	\begin{minipage}{\textwidth}
		\begin{minipage}{0.5\textwidth}
			\begin{prooftree}
				\AxiomC{$\ptribothil{p_1}{t}{q_1}$}
				\AxiomC{$\ptribothil{p_2}{t}{q_2}$}
				\RightLabel{(disj)}
				\BinaryInfC{$\ptribothil{p_1 + p_2}{t}{q_1 + q_2}$}
			\end{prooftree}
		\end{minipage}
		\hfill
		\begin{minipage}{0.425\textwidth}
			\begin{prooftree}
				\AxiomC{$\ptrierril{p}{t_{\text{$1$}}}{q}$}
				\RightLabel{(short-circuit)}
				\UnaryInfC{$\ptrierril{p}{t_{\text{$1$}} \cdot t_{\text{$2$}}}{q}$}
			\end{prooftree}
		\end{minipage}
	\end{minipage}
	\\[3pt]
	\begin{minipage}{\textwidth}
		\begin{minipage}{0.6\textwidth}
			\begin{prooftree}
				\AxiomC{$\ptriokil{p}{t_{\text{$1$}}}{r}\!\!$}
				\AxiomC{$\!\!\ptribothil{r}{t_{\text{$2$}}}{q}$}
				\RightLabel{(seq-normal)}
				\BinaryInfC{$\ptribothil{p}{t_{\text{$1$}} \cdot t_{\text{$2$}}}{q}$}
			\end{prooftree}
		\end{minipage}
		\hfill
		\begin{minipage}{0.36\textwidth}
		\begin{prooftree}
		\AxiomC{\phantom{$\ptriokil{p}{t_1}{r}$}}
				\RightLabel{(iterate zero)}
				\UnaryInfC{$\ptriokil{p}{t^*}{p}$}
		\end{prooftree}
		\end{minipage}
	\end{minipage}
	\\[3pt]
	\begin{minipage}{\textwidth}
		\begin{minipage}{0.5\textwidth}
			\begin{prooftree}
				\AxiomC{$\forall n \in \mathbb{N} . \; \ptriokil{p_n}{t}{p_{n+1}}$}
				\RightLabel{(back-v)}
				\UnaryInfC{$\ptriokil{p_0}{t^*}{\disj p_n}$}
			\end{prooftree}
		\end{minipage}
		\hfill
		\begin{minipage}{0.5\textwidth}
			\begin{prooftree}
				\AxiomC{$\ptribothil{p}{t^* \cdot t}{q}$}
				\RightLabel{(iterate non-zero)}
				\UnaryInfC{$\ptribothil{p}{t^*}{q}$}
			\end{prooftree}
		\end{minipage}
	\end{minipage}
	\\[3pt]
	\begin{minipage}{\textwidth}
		\begin{minipage}{\textwidth}
			\begin{prooftree}
				\AxiomC{$\ptribothil{p}{t_{\text{$i$}}}{q}$, with $i \in \{1,2\}$}
				\RightLabel{(choice)}
				\UnaryInfC{$\ptribothil{p}{t_{\text{$1$}} + t_{\text{$2$}}}{q}$}
			\end{prooftree}
		\end{minipage}
	\end{minipage}
				\end{framed}
	\caption{Proof system $\IL$.}
	\label{fig:il}
\end{figure}

Section \ref{sec:ual} has shown that $\LCK$ yields a generalization of $\UL$. The same can be done for $\IL$, i.e., we prove 
that $\LCIL_A$ with incorrectness triples generalizes the incorrectness logic of \cite{ohearn_incorrectness_2020}.
For the sake of clarity, we recall in Figure \ref{fig:il} an algebraic version of $\IL$.
Analogously to the reduction of Theorem~\ref{thm:lcl_ul_equiv}, this generalization is 
obtained by letting $A = A_{tr}$, where $A_{tr}$ is the trivial abstract domain.

\begin{restatable}{theorem}{thmlclilequiv} \label{thm:lcil_il_equiv}
Let $K$ be a CTC bdKAT and $T_{\Sigma,B}$ a language interpreted on $K$. For any $p,q \in \Test(K)$, $\mathtt{t} \in T_{\Sigma, B}$, 
	\[ \ptriokerrtr{p}{t}{q}{r} \quad \Leftrightarrow \quad \ptriokerril{p}{t}{q}{r}. \]
\end{restatable}
\begin{proof}
	Notice that the \ok part of $\LCIL$ coincides with $\LCK$\footnote{With the additional rule (limit) and without the rules (rec), (iterate).} and likewise $\IL$ coincides with $\UL$. Hence this part of the equivalence follows by Theorem \ref{thm:lcl_ul_equiv}. We only need to show the equivalence of the \err part.
	
	\vspace{2px} \noindent $\Rightarrow)$ Assuming $\ptrierrtr{p}{t}{r}$ we need to show $\ptrierril{p}{t}{r}$. The proof is by induction on the derivation tree.
	
	\vspace{2px} \noindent (tranfer)\,: Coincides with (transfer). \\
	\noindent (relax)\,: Coincides with (consequence). \\
	\noindent (seq-ok)\,: Is the \ok part of (seq-normal). \\
	\noindent (seq-err)\,: Can be derived as follows
	\begin{scprooftree}{0.85}
		\AxiomC{$\ptriokil{p}{t_1}{q}$}		
		\AxiomC{$\ptrierril{q}{t_2}{s}$}
		\RightLabel{(seq-normal)}
		\BinaryInfC{$\ptrierril{p}{t_1 \cdot t_2}{s}$}
		\AxiomC{$\ptrierril{p}{t_1}{r}$}
		\RightLabel{(short-circuit)}
		\UnaryInfC{$\ptrierril{p}{t_1 \cdot t_2}{r}$}
		\RightLabel{(disj)}
		\BinaryInfC{$\ptrierril{p}{t_1 \cdot t_2}{s + r}$}
	\end{scprooftree}
	\noindent (rec-err)\,: Can be derived as follows
	\begin{prooftree}
		\AxiomC{$\ptriokil{p}{t^*}{q}$}
		\AxiomC{$\ptrierril{q}{t}{r}$}
		\RightLabel{(seq-normal)}
		\BinaryInfC{$\ptrierril{p}{t^* \cdot t}{r}$}
		\RightLabel{(iterate non-zero)}
		\UnaryInfC{$\ptrierril{p}{t^*}{r}$}
	\end{prooftree}
	\noindent (join)\,: As in Theorem \ref{thm:lcl_ul_equiv}. \\
	\noindent (limit)\,: Coincides with (back-v).	
	
	\vspace{5px} \noindent $\Leftarrow)$ Assuming $\ptrierril{p}{t}{r}$ we need to show $\ptrierrtr{p}{t}{r}$. As in Theorem \ref{thm:lcl_ul_equiv}, we exploit the completeness of $\LCIL$ and likewise we observe that the condition (2) is trivially true as $\asemaerr{t} \alpha(p) = \top = \alpha(q) = \top = \alpha(\csemaerr{t} p)$, meaning that the only condition to check is (1). The proof is by structural induction on the derivation tree of $\ptrierril{p}{t}{r}$.
	
	\vspace{3px} \noindent (transfer)\,: The premises $\mathbb{C}^A_p(\csemaerr{a})$ and $\mathbb{C}^A_p(\csemaok{a})$ hold for any $\mathtt{a}$ and $p$ because $A_{tr}(\cdot) = 1$, so that (transfer) yields $\ptrierrtr{p}{a}{\csemaerr{a} p}$.
	
	\vspace{3px} \noindent (empty)\,: For any $\mathtt{t} \in T_{\Sigma, B}$ it holds $0 \leq \csemaerr{t} p$, meaning $\vtrierrtr{p}{t}{0}$ and by completeness of $\LCIL$ we have $\ptrierrtr{p}{t}{0}$.

	\vspace{3px} \noindent (consequence)\,: The inequalities $p \leq A_{tr}(p')$ and $q' \leq A_{tr}(q)$ are always true because $A_{tr}(\cdot) = 1$ which is the greatest element of $\Test(K)$ meaning that the rule coincides with (relax).
	
	\vspace{3px} \noindent (disj) and (choice)\,: Same as in Theorem \ref{thm:lcl_ul_equiv}.

	\vspace{3px} \noindent (short-circuit)\,: By induction $\ptrierrtr{p}{t_1}{r}$ and by soundness $\vtrierrtr{p}{t_1}{r}$. By (\ref{eqn:csema_err_mul}) $\csemaerr{t_1 \cdot t_2} p = \csemaerr{t_1} p + \csemaerr{t_2} \csemaok{t_1} p \geq \csemaerr{t_1} p \geq r$
	which means $\vtrierrtr{p}{t_1 \cdot t_2}{r}$ and by completeness $\ptrierrtr{p}{t_1 \cdot t_2}{r}$.

	\vspace{3px} \noindent (seq-normal)\,: By induction and soundness we have $r \leq \csemaok{t_1} p$ and $q \leq \csemaerr{t_2} r$. By (\ref{eqn:csema_err_mul}) $\csemaerr{t_1 \cdot t_2} p = \csemaerr{t_1} p + \csemaerr{t_2} \csemaok{t_1} p \geq  \csemaerr{t_2} \csemaok{t_1} p \geq \csemaerr{t_2} r \geq q$, which means $\vtrierrtr{p}{t_1 \cdot t_2}{q}$ and by completeness $\ptrierrtr{p}{t_1 \cdot t_2}{q}$.
	
	\vspace{3px} \noindent (iterate non-zero)\,: By induction and soundness 
	\[ q \leq \csemaerr{t^* \cdot t} p \stackrel{(\ref{eqn:csema_err_mul})}{=} \csemaerr{t^*} p + \csemaerr{t} \csemaok{t^*} p \stackrel{(\ref{eqn:csema_err_star})}{=} \csemaerr{t^*} p \]
	which means $\vtrierrtr{p}{t^*}{q}$ and by completeness $\ptrierrtr{p}{t^*}{q}$. \qed
\end{proof}

The abstraction map $\alpha=\lambda x.\top$ of $A_{tr}$ makes the validity of a triple trivially true. In particular,
\( \asematrok{t} \alpha(p) = \top = \alpha(q) = \alpha(\csemaok{t} p) \)
and
\( \asematrerr{t} \alpha(p) = \top = \alpha(q) = \alpha(\csemaerr{t} p) \) hold. 
As a consequence, we obtain that
\begin{align}
\vtriokerrtr{p}{t}{q}{r} \quad \Leftrightarrow \quad\: \vtriokerril{p}{t}{q}{r} \label{eq:Atr-IL}
\end{align}
By this equivalence \eqref{eq:Atr-IL} and Theorems \ref{thm:lsound_il} and~\ref{thm:lcompl_il}, we can thus 
retrieve the logical soundness and completeness of $\IL$ as a consequence of the one of $\LCIL_{A_{tr}}$.

\begin{corollary}
Let $K$ be a CTC bdKAT and $T_{\Sigma,B}$ a language interpreted on $K$.
For any $p,q \in \Test(K)$, $\mathtt{t} \in T_{\Sigma, B}$,
	\(\ptriokerril{p}{t}{q}{r} \; \Leftrightarrow \;\: \vtriokerril{p}{t}{q}{r}\).
\end{corollary}

\section{Local Completeness Logic in TopKAT} \label{sect:lctk}
We have shown in Section~\ref{sec:lclkat} how KAT extended with a modal  backward-diamond operator allows us to interpret and
represent the local completeness program logic. This result follows the approach 
by Moller, O'Hearn and Hoare~\cite{fahrenberg_algebra_2021}, who leverage a backward-diamond operator in their KAT interpretation
of correctness/incorrectness logics. 
On the other hand, Zhang, de Amorim and 
Gaboardi~\cite{zhang_incorrectness_2022} have recently shown that incorrectness logic can be formulated 
for a standard KAT, provided that it contains a top element, thus giving rise to a so-called TopKAT. In particular, \cite{zhang_incorrectness_2022}  observed that a TopKAT is enough to express the codomain of relational KATs. 
In this section, we take a similar path in studying an alternative formulation of local completeness logic based on a TopKAT.

\subsection{Abstracting TopKATs} \label{sect:top_abstr_dom}
We expect that the base case of abstract semantics 
$\asema{\mathtt{a}} p^\sharp$ for a basic action $\mathtt{a} \in \Atom$ is 
defined as best correct approximation in $A$ of the concrete semantics of $\mathtt{a}$ on the concretization of $p^\sharp$.  
In a bdKAT this is achieved in definition 
\eqref{def:abs-sem} through its backward-diamond operator, which is crucially used in \eqref{bd-csem} to define the 
strongest postcondition as $\csema{a}^K \gamma(p^\sharp) = \bdia{\intrp{a}} \gamma(p^\sharp)$.
Zhang et al.~\cite{zhang_incorrectness_2022} observed that in a relational model of KAT, the codomain inclusion $\cod(q) \subseteq \cod(pa)$ defining the meaning of an 
under-approximation triple $\otri{p}{a}{q}$ can be expressed in a TopKAT as the inequality 
$\top q \leq \top pa$, thus hinting that this latter condition could be taken as definition of validity of incorrectness triples in a TopKAT. We follow here a similar approach by considering the element $\top p \intrp{a}$ as a proxy for strongest postconditions in a TopKAT. 
It is worth noticing that while in a bdKAT a strongest postcondition $\bdia{\intrp{a}}p$ is always a test, in a TopKAT $K$, 
given $p \in \Test(K)$ and a term $\mathtt{t} \in T_{\Sigma, B}$,  
it is not guaranteed that there exists a test $q \in \Test(K)$ such that $\top p \intrp{t} = \top q$, as shown by the following example.

\begin{examplebf}[Strongest Postconditions in TopKAT] \label{exa:top_no_compl}
	Consider the Kleene algebra $A_3 = \{0,1,a\}$ consisting of 3 elements and characterized by Conway~\cite[Chapter~12]{conway2012regular}. This algebra can be lifted to a KAT by letting $\Test(A_3) \ud \{0,1\}$
	and defining the KAT operators as follows: 

		\smallskip
		\begin{center}
			$\begin{array}{|c|ccc|}
			\hline
				~+~ & ~0~ & ~1~ & ~a~\\
				\hline 
				0 & 0 & 1 & a \\
				1 & 1 & 1 & 1 \\ 
				a & a & 1 & a \\
				\hline
			\end{array}$
			\qquad\qquad
			$\begin{array}{|c|ccc|}
							\hline 
				~\cdot~ & ~0~ & ~1~ & ~a~\\
				\hline 
				0 & 0 & 0 & 0 \\
				1 & 0 & 1 & a \\ 
				a & 0 & a & 0 \\
								\hline 
			\end{array}$
			\qquad\qquad
			$\begin{array}{|c|c|c|}
							\hline 
							 \rule{0pt}{\normalbaselineskip}
				~\mathord{0^*}\ud 1~ & ~\mathord{1^*}\ud 1 & ~\mathord{a^*}\ud 1~
				\rule{0pt}{\normalbaselineskip}
				\\							 
				\hline 
			\end{array}$
		\end{center}

	\smallskip
	\noindent 
	We have that $1 \geq a$ and $1 \geq 0$, because $1+a=1$ and $1+0=1$, so that $A_3$ is a TopKAT with $\top = 1$. Moreover, $\top \cdot 1 \cdot a = 1 \cdot 1 \cdot a = a$, whereas there exists no $q \in \Test(A_3)$ satisfying $\top \cdot q = a$. Indeed, $\top \cdot 1 = 1 \cdot 1 = 1 \neq a$ and $\top \cdot 0 = 0 \neq a$.
	\qed
\end{examplebf}

In general, the lack of such a $q \in \Test(K)$ implies that the abstract domain cannot be defined as an abstraction of  
the set of topped-tests $\{\top p \mid p \in \Test(K) \}$, because in this case we could miss the abstraction 
$\alpha(\top p \intrp{a})$. To settle this issue, an abstract domain must provide an approximation of the larger set
\begin{equation*}
\topp(K) \ud \{ \top a \mid a \in K \}
\end{equation*}
which contains all the multiplicative elements of type $\top a$. %

\begin{definitionrm}[Top Kleene Abstract Domain]\label{def:tkad}
	A poset $(A, \leq)$ is a \emph{top Kleene abstract domain} of a TopKAT $K$ if: 
	\begin{enumerate}[{\rm (i)}]
	\item There exists a Galois insertion, defined by  $\gamma : A \rightarrow \topp(K)$ and $\alpha : \topp(K) \rightarrow A$, 
		of the poset $(A,\leq_A)$ into the poset 
		$(\topp(K),\leq_K)$;  
		\item $A$ is countably-complete. \qed
		\end{enumerate}
\end{definitionrm}

The abstract semantic  function
\( \asema{\cdot} : T_{\Sigma, B} \rightarrow (A \rightarrow A) \)
 on a top Kleene abstraction $A$ 
can be therefore defined 
for the base case
$\mathtt{a}\in \Atom$ as
\[\asema{a} p^\sharp \triangleq \alpha(\gamma(p^\sharp) \intrp{a}),\]
while the remaining inductive cases are defined 
as in \eqref{def:abs-sem} for Kleene abstractions. 
The monotonicity and soundness properties of this abstract semantics hold, provided that 
the TopKAT is  $^*$-continuous\footnote{This condition plays a role similar to the CTC condition for bdKATs.}, which is referred to as 
TopKAT$^*$.
\begin{restatable}[Soundness of TopKAT Abstract Semantics]{theorem}{propasematop}
Let $A$ be a Kleene abstraction of a TopKAT$^*$ $K$
and $T_{\Sigma,B}$ be a language interpreted on $K$. For all $p^\sharp, q^\sharp \in A$, $a \in K$ and $\mathtt{t} \in T_{\Sigma, B}$:
	\begin{align*}
		\tag{monotonicity} p^\sharp \leq_A q^\sharp \Rightarrow \asema{t} p^\sharp \leq_A \asema{t} q^\sharp \\
		\tag{soundness} \alpha(\top a \intrp{t}) \leq_A \asema{t} \alpha(\top a)
	\end{align*}
\end{restatable}
\begin{proof}
	The first property can be proved by induction on the structure of $\mathtt{t} \in T_{\Sigma, B}$. Notice that the inductive cases are the same as Theorem \ref{prop:prop_asema}, therefore we only need to verify the base case.
	\[ \asema{a} p^\sharp = \alpha(\gamma(p^\sharp) \intrp{a}) \leq \alpha(\gamma(q^\sharp) \intrp{a})   = \asema{a} q^\sharp \]
	where we used the isotony of $\alpha(\cdot)$, $\gamma(\cdot)$ and $\cdot$.
	
	\noindent Likewise we prove the second property by induction on the structure of ${\mathtt{t} \in T_{\Sigma, B}}$. The computation can be easily adapted from Theorem $\ref{prop:prop_asema}$. As an example we explicitly prove the $\mathtt{t_0^*}$ case.

	\vspace{5px} \noindent($\mathtt{t_0^*}$)\,: In order to prove the result we need a preliminary fact
	\begin{equation} \label{eqn:claim_4}
		\alpha(\top a (\intrp{t_0})^n) \leq (\asema{t_0})^n \alpha(\top a)
	\end{equation}
	which can be shown by induction on $n$. The base case holds as $\alpha(\top a (\intrp{t_0})^0) = \alpha(\top a) = (\asema{t_0})^0 \alpha(\top a)$. The inductive case can be proved as follows
	\[ \alpha(\top a (\intrp{t_0})^{n+1}) = \alpha(\top a (\intrp{t_0})^n \intrp{t_0}) \leq \asema{t_0} \alpha(\top a (\intrp{t_0})^n) \leq (\asema{t_0})^{n+1} \alpha(\top a) \]
	The previous fact and the \textit{$^*$-continuity} of $K$ yield the soundness
	\begin{align*}
		\alpha(\top a\intrp{t_0^*}) &= \\
		\alpha(\top a (\intrp{t_0})^*) &= \textrm{[By $^*$-continuity]} \\
		\alpha(\disj \top a (\intrp{t_0})^n) &= \textrm{[By additivity of $\alpha(\cdot)$]} \\
		\disj \alpha(\top a (\intrp{t_0})^n) &\leq \textrm{[By (\ref{eqn:claim_4})]} \\
		\disj (\asema{t_0})^n \alpha(\top a) &= \\
		\asema{t_0^*} \alpha(\top a). \tag*{\qed}
	\end{align*} 
\end{proof}

\subsection{Local Completeness Logic on TopKAT}
Completeness and triple validity are adapted to the TopKAT framework as follows.
Given a Top Kleene abstract domain $A$ on a TopKAT$^*$ $K$, 
$A$ is defined to be locally complete for $a\in K$ on an element $b \in K$, denoted by 
$\mathbb{C}^A_b(a)$, when
		\[ A(\top b a) = A(A(\top b) a) \]
holds. Moreover, $A$ is globally complete for $a$, denoted by $\mathbb{C}^A(a)$, when it is locally complete for any $b \in K$.

Likewise, a triple $\tri{a}{t}{b}$, with $a,b \in K$ and $\mathtt{t} \in T_{\Sigma, B}$, 
is valid, denoted by ${\vtrit{a}{t}{b}}$, when: 
\[
\text{(1)~~$\top b \leq \top a \intrp{t}$;\quad
			(2)~~$\asema{t} \alpha(\top a) = \alpha(\top b) = \alpha(\top a \intrp{t})$. 
			}
\]
			
The corresponding proof system, denoted by $\LCTK_A$, has the same rules 
of $\LCK_A$ in Figure \ref{fig:lcl} except (transfer),
(relax) and (iterate) which are modified as follows:  
\begin{center}
{\small
			\begin{prooftree}
				\AxiomC{$\mathtt{c} \in \Atom$}
				\AxiomC{$\mathbb{C}^A_a(\intrp{c})$}
				\RightLabel{(transfer)}
				\BinaryInfC{$\ptrit{a}{c}{a \intrp{c}}$}
			\end{prooftree}
			\begin{prooftree}
				\AxiomC{$\top a' \leq \top a \leq A(\top a') \qquad \ptrit{a'}{t}{b'} \qquad \top b \leq \top b' \leq A(\top b)$}
				\RightLabel{(relax)}
				\UnaryInfC{$\ptrit{a}{t}{b}$}
			\end{prooftree}
			\begin{prooftree}
				\AxiomC{$\ptrit{a}{t}{b}$}
				\AxiomC{$\top b \leq A(\top a)$}
				\RightLabel{(iterate)}
				\BinaryInfC{$\ptrit{a}{t^*}{a + b}$}
			\end{prooftree}
}
\end{center}

This incarnation $\LCTK_A$ of local completeness logic for 
TopKAT$^*$ turns out to be logically sound and, under additional hypotheses, complete. 

\begin{restatable}[Logical Soundness of $\vdash^{{\scriptscriptstyle\mathrm{TK}}}_{A}$]{theorem}{thmlsoundtop} \label{thm:lsound_top}
	If $\ptrit{a}{t}{b}$ then
	\begin{enumerate}[{\rm (i)}]
		\item $\top b \leq \top a \intrp{t}$;
		\item $\asema{t} \alpha(\top a) = \alpha(\top b) = \alpha(\top a \intrp{t})$.
	\end{enumerate}
\end{restatable}
\begin{proof}
	The proof is as in Theorem \ref{thm:lsound}. As an example we explicitly prove the cases (transfer) and (iterate).
%

	\vspace{5px} \noindent (transfer)\,: (i) is immediate because $\top b = \top a \intrp{c} \leq \top a \intrp{c}$. (ii) instead is a consequence of local completeness ${\asema{c} \alpha(\top a) = \alpha(\gamma(\alpha(\top a)) \intrp{c}) \stackrel{\text{LC}}{=} \alpha(\top a \intrp{c}) = \alpha(\top b)}$.

	\vspace{5px} \noindent (iterate)\,: For any $k \in K$ it holds $1 \leq k^*$ and $k \leq k^*$ \cite[Section 2.1]{KOZEN1994366}. In particular $\top b \leq \top a \intrp{t} \leq \top a (\intrp{t})^*$ and $\top a = \top a \cdot 1 \leq \top a (\intrp{t})^*$, thus implying (i) ${\top (a+b)\leq \top a (\intrp{t})^* = \top a \intrp{t^*}}$.	To show the (ii) we need a preliminary fact
	\begin{equation} \label{eqn:claim_8}
		\asema{t} \alpha(\top a) \leq \alpha(\top a) \Rightarrow \disj (\asema{t})^n \alpha(\top a)  = \alpha(\top a)
	\end{equation}
	As a first step, let us prove by induction on $n$ that $\alpha(\top a)$ is an upper-bound of the elements of the disjunction, i.e., $\forall n \in \mathbb{N} . \; (\asema{t})^n \alpha(\top a) \leq \alpha(\top a)$. The base case is trivially true $(\asema{t})^0 \alpha(\top a) = \alpha(\top a) \leq \alpha(\top a)$, while the inductive case can be shown as follows: ${(\asema{t})^{n+1} \alpha(\top a) = (\asema{t})^n \asema{t} \alpha(\top a) \leq (\asema{t})^n \alpha(\top a) \leq \alpha(\top a)}$.
	Notice that $\alpha(\top a)$ is also the least upper bound because it is part of the disjunction as $(\asema{t})^0 \alpha(\top a) = \alpha(\top a)$. The hypothesis $\top b \leq A(\top a)$ can be rewritten as $\alpha(\top b) \leq \alpha(\top a)$ by isotony of $\alpha(\cdot)$,$\gamma(\cdot)$, (\ref{eqn:gamma_inject}), (\ref{eqn:uco_idem}). Finally, $\asema{t} \alpha(\top a) = \alpha(\top b) \leq \alpha(\top a)$, meaning that the premise of (\ref{eqn:claim_8}) holds, so that we can conclude
	\[ \asema{t^*} \alpha(\top a) =  \disj (\asema{t})^n \alpha(\top a) = \alpha(\top a) = \alpha(\top a) + \alpha(\top b) = \alpha(\top (a + b)) \] \qed
\end{proof}

\noindent 
Logical completeness needs the following additional conditions:
\begin{enumerate}[{\rm (a)}]
	\item\label{tk1} Likewise $\LCK_A$, the same infinitary rule for Kleene star:
	\begin{prooftree}
		\AxiomC{$\forall n \in \mathbb{N} . \; \ptrit{a_n}{t}{a_{n+1}}$}
		\RightLabel{(limit)}
		\UnaryInfC{$\ptrit{a_0}{t^*}{\disj a_n}$}
	\end{prooftree}
	where we assume that:
	\begin{itemize}
		\item $\disj a_n$ always exists. Let us remark that for bdKAT, 
		such explicit condition was not needed, as it was entailed by the CTC requirement on the KAT.
		\item $\top$ distributes over $\disj a_n$, i.e., 
		\(\top \disj a_n = \disj \top a_n \).
	\end{itemize}
	\item\label{tk2} Global completeness of all the primitive actions and tests occurring in the program.
\end{enumerate}

As in $\LCK_A$, (limit) is sound:
\begin{restatable}{lemma}{lmlsoundlimittop}\label{lemma:lmlsoundlimittop}
	With the hypothesis of Theorem \ref{thm:lsound_top}, the rule (limit) is logically sound.
\end{restatable}
\begin{proof}
	Instantiating (i) we obtain $\top \disj a_n \leq \top a_0 \intrp{t^*}$. The distributivity condition of (limit) yields $\top \disj a_n = \disj \top a_n$. For this reason (i) is equivalent to ${\forall n \in \mathbb{N} . \; \top a_0 \intrp{t^*} \geq \top a_n}$. This fact can be shown by induction on $n$.

	\vspace{5px} \noindent For $n=0$\,: $\top a_0 \intrp{t^*} = \top a_0 (\intrp{t})^* \geq \top a_0$, as for any $k \in K$, $k^* \geq 1$ by (\ref{eqn:star_unfold}).
	
	\vspace{5px} \noindent For $n+1$\,: The rule has a premise $\ptrit{a_n}{t}{a_{n+1}}$, and by induction (on the derivation tree) we have that $\vtrit{a_n}{t}{a_{n+1}}$.
		\begin{align*}
			\top a_0 \intrp{t^*} &= \\
			\top a_0 (\intrp{t})^* &\geq \textrm{[By (\ref{eqn:star_unfold})]} \\
			\top a_0 (1 + (\intrp{t})^* \intrp{t}) &\geq \\
			\top a_0 (\intrp{t})^* \intrp{t} &= \\
			\top a_0 \intrp{t^*} \intrp{t} &\geq \\
			\top a_n \intrp{t} &\geq \textrm{[By $\vtrit{a_n}{t}{a_{n+1}}$]} \\
			\top a_{n+1}
		\end{align*}
	
	\noindent For condition (ii) instead we need to show $\asema{t^*} \alpha(\top a_0) = \alpha(\top \disj a_n)$, provided that $\asema{t} \alpha(\top a_n) = \alpha(\top a_{n+1})$.
	By definition, $\asema{t^*} \alpha(\top a_0) = \disj (\asema{t})^n \alpha(\top a_0)$ and we have also that $\alpha(\top \disj a_n) = \alpha(\disj \top a_n) = \disj \alpha(\top a_n)$. Therefore (ii) is implied by
	\[\forall n \in \mathbb{N} . \; (\asema{t})^n \alpha(\top a_0) = \alpha(\top a_n) \]
	That can be show by induction on $n$. The base case is trivial $(\asema{t})^0 \alpha(\top a_0) = \alpha(\top a_0)$, while the inductive can be shown as follows
	\begin{align*} 
	(\asema{t})^{n+1} \alpha(\top a_0) = \asema{t} (\asema{t})^n \alpha(\top a_0) = \asema{t} \alpha(\top a_n) = 
	\alpha(\top a_{n+1}). 
    \tag*{\qed}
	\end{align*} 
\end{proof}

Condition \eqref{tk2} entails global completeness:
\begin{restatable}{lemma}{lmglobcompltop} \label{lemma:glob_compl_top}
	Let $K$ a TopKAT$^*$, $A$ a top Kleene abstract domain on $K$ and $T_{\Sigma,B}$ a KAT language on $K$. For any $\mathtt{t} \in T_{\Sigma, B}$ and $a \in K$ we have
	\[ (\forall \mathtt{b} \in \Atom(\mathtt{t}) . \; \mathbb{C}^A(\intrp{b})) \Rightarrow \asema{t} \alpha(\top a) = \alpha(\top a \intrp{t}) \]
\end{restatable}
\begin{proof}
	The proof is as in Lemma \ref{lemma:glob_compl}. We only prove explicitly the $\mathtt{t_0^*}$ case.
	
%
%
	\vspace{5px} \noindent $(\mathtt{t_0^*})$\,: The result is a consequence of the following claim
		\begin{equation} \label{eqn:claim_5}
			\forall n \in \mathbb{N}. \; (\asema{t_0})^n \alpha(\top a) = \alpha(\top a (\intrp{t_0})^n)
		\end{equation}
		that can be proved by induction. The base case is as follows: 
		$(\asema{t_0})^0 \alpha(\top a) = \alpha(\top a) = \alpha(\top a (\intrp{t_0})^0)$. The inductive case can be proved as follows
		\begin{align*}
			(\asema{t_0})^{n+1} \alpha(\top a) &= \\
			\asema{t_0} (\asema{t_0})^n \alpha(\top a) &= \\
			\asema{t_0} \alpha(\top a (\intrp{t_0})^n) &= \\
			\alpha(\top a (\intrp{t_0})^n \intrp{t_0}) &= \\
			\alpha(\top a (\intrp{t_0})^{n+1})
		\end{align*}
		Moreover by $^*$-continuity we have that $\top a \intrp{t_0^*} = \top a (\intrp{t_0})^* = \disj \top a (\intrp{t_0})^n$, meaning that
		\begin{align*}
			\asema{t_0^*} \alpha(\top a) &= \\
			\disj (\asema{t_0})^n \alpha(\top a) &= \textrm{[By (\ref{eqn:claim_5})]} \\
			\disj \alpha(\top a (\intrp{t_0})^n) &= \textrm{[Additivity of $\alpha(\cdot)$)]} \\
			\alpha(\disj \top a (\intrp{t_0})^n) &= \\
			\alpha(\top a \intrp{t_0^*}) \tag*{\qed}
		\end{align*} 
\end{proof}

\begin{restatable}[Logical Completeness of $\vdash^{{\scriptscriptstyle\mathrm{TK}}}_{A}$]{theorem}{thmlcompltop} \label{thm:lcompl_top}
Assume that conditions \eqref{tk1} and \eqref{tk2} hold.
	If $\vtrit{a}{t}{b}$ then $\ptrit{a}{t}{b}$.
\end{restatable}
\begin{proof}
	The proof is as in Theorem \ref{thm:lcompl}. The case $(\mathtt{t_0^*})$ requires special attention, and is proved explicitly.

%
	\vspace{5px} \noindent ($\mathtt{t_0^*}$)\,: Let $a_n = a_0 (\intrp{t_0})^n$. As a first step we prove that ${\forall n \in \mathbb{N} . \; \vtrit{a_n}{t_0}{a_{n+1}}}$: condition (i) holds because $\top a_{n+1} = \top a_0 (\intrp{t_0})^{n+1} = \top a_0 (\intrp{t_0})^n \intrp{t_0} = \top a_n \intrp{t_0}$, while (ii) holds by Lemma \ref{lemma:glob_compl_top}, $\asema{t_0} \alpha(\top a_n) = \alpha(\top a_n \intrp{t_0}) = \alpha(\top a_{n+1})$. Finally, the inductive hypothesis yields $\forall n \in \mathbb{N} . \; \ptrit{a_n}{t_0}{a_{n+1}}$.

	We want to use the rule (limit) with postcondition $\disj a_n$. To do so, we need to show two additional facts: 1) the disjunction exists 2) $\top$ distributes over the disjunction. 1) holds as $a_0 (\intrp{t_0})^* = \disj a_0 (\intrp{t_0})^n = \disj a_n$, where the first equality holds by $^*$-continuity. By $^*$-continuity we have that $\top a_0 (\intrp{t_0})^* = \disj \top a_0 (\intrp{t_0})^n = \disj \top a_n$ and $\top a_0 (\intrp{t_0})^* = \top \disj a_0 (\intrp{t_0})^n = \top \disj a_n$, thus implying $\disj \top a_n = \top \disj a_n$. Finally, by (limit) we obtain the result $\ptrit{a}{t_0^*}{a \intrp{t_0^*}}$. \qed
%
\end{proof}

Let us describe an example of derivation in $\LCTK_A$. 

\begin{examplebf}
	Consider a relational KAT $K=\wp(\bZ\times \bZ)$ on the set of integers $\mathbb{Z}$, where $\mathbf{1}_K \triangleq \{(z,z) \mid z \in \mathbb{Z} \}$ and $\mathbf{0}_K \triangleq \varnothing$. Notice that $\bZ\times \bZ \in K$ is the top element $\top$ of $K$, meaning that $K$ is a TopKAT. Let us consider a language with primitive actions $\Sigma = \{ \mathtt{x:= x+1} \}$ and primitive tests $B = \{ \mathtt{x \geq 0}, \mathtt{x < 0} \}$. The evaluation function ${u : \Sigma \cup B \ra K}$ is defined as expected by the following relations:
	\begin{align*}
		&u(\mathtt{x := x+1}) \ud \{(z,z+1) \mid z \in \mathbb{Z} \},\\ 
		& u(\mathtt{x \geq 0}) \ud \{(z,z) \mid z \in \mathbb{Z}, z \geq 0 \}, \\
		&u(\mathtt{x < 0}) \ud \{(z,z) \mid z \in \mathbb{Z}, z < 0 \}.
	\end{align*}
	Consider the following sign abstraction $\Sign\ud \{\bZ, \bZ_{\leq 0}, \bZ_{\neq 0}, \bZ_{\geq 0}, \bZ_{<0}, \bZ_{=0}, \bZ_{>0}, 
	\varnothing \}$ of $\wp(\mathbb{Z})$, whose abstraction and concretization maps are straightforward. 
	Let us verify that the program \[{\mathtt{r} \equiv \big((\mathtt{x \geq 0}) \cdot (\mathtt{x := x+1})\big)^* \cdot (\mathtt{x < 0})}\] does not terminate with precondition $p \triangleq \{(0,0), (10,10)\}$, i.e., we prove the specification $\Spec \triangleq \varnothing$. Let us define the following auxiliary elements: $q \ud \{(1,1),(11,11)\} $, ${s \ud p + q}$, ${t_{\geq 0} \ud \{(x,z) \mid x\in \mathbb{Z}, z\in \bZ_{\geq 0} \}}$, and observe that $\Sign(t_{\geq 0})= t_{\geq 0}$. 
	
	\noindent
	The following local completeness conditions for the atoms hold: 
	\begin{align*}
		\alpha(\Sign(\top p) \intrp{x \geq 0}) &= \alpha(t_{\geq 0} \intrp{x \geq 0}) = \mathbb{Z}_{\geq 0} = \alpha(\top p \intrp{x \geq 0}), \\
		\alpha(\Sign(\top p) \intrp{x := x+1}) &= \alpha(t_{\geq 0} \intrp{x := x+1}) = \mathbb{Z}_{> 0} = \alpha(\top p \intrp{x:= x+1}), \\
		\alpha(\Sign(\top s) \intrp{x < 0}) &= \alpha(t_{\geq 0} \intrp{x < 0}) = \varnothing = \alpha(\top s \intrp{x < 0}).
			\end{align*}
Moreover, we also have that:
	\begin{align*}
		\top q &= \{(x,z) \mid x \in \mathbb{Z}, z \in \{1,11\} \} \leq \: t_{\geq 0} = \Sign(\top p).
	\end{align*}
	
	\noindent The following derivation shows that the triple $[p] \;\mathtt{r}\; [\mathbf{0}_K]$
	is provable in $\LCTK_{\Sign}$:
	\begin{scprooftree}{0.63}
		\AxiomC{$\mathbb{C}^{\Sign}_p(\intrp{x \geq 0})$}
		\RightLabel{(transfer)}
		\UnaryInfC{$\patrit{p}{x \geq 0}{p}{\Sign}$}
		\AxiomC{$\mathbb{C}^{\Sign}_p(\intrp{x := x+1})$}
		\RightLabel{(transfer)}
		\UnaryInfC{$\patrit{p}{x := x+1}{q}{\Sign}$}
		\RightLabel{(seq)}
		\BinaryInfC{$\patrit{p}{(x \geq 0) \cdot (x := x+1)}{q}{\Sign}$}
		\AxiomC{$\top q \leq \Sign(\top p)$}
		\RightLabel{(iterate)}
		\BinaryInfC{$\patrit{p}{\big((x \geq 0) \cdot (x := x+1)\big)^*}{s}{\Sign}$}
		\AxiomC{$\mathbb{C}^{\Sign}_s(\intrp{x < 0})$}
		\RightLabel{(transfer)}
		\UnaryInfC{$\patrit{s}{x < 0}{\mathbf{0}_K}{\Sign}$}
		\RightLabel{(seq)}
		\BinaryInfC{$\patrit{p}{\big((x \geq 0) \cdot (x := x+1)\big)^* \cdot (x < 0)}{\mathbf{0}_K}{\Sign}$}
\end{scprooftree}
	
	\medskip
	\noindent
	By Theorem \ref{thm:lsound_top}, we have that $\top \mathbf{0}_K \subseteq \top p \intrp{r} \subseteq \Sign(\top \mathbf{0}_K) = \varnothing$, meaning that the program does not terminate, and $\Spec$ is satisfied as $\top p \intrp{r} = \varnothing = \top \Spec$.
\qed
\end{examplebf}

\subsection{Relationship with Under-Approximation Logic} \label{sect:lctk_ul}
We have shown in Section \ref{sec:ual} that the backward-diamond formulation of $\LCK$ generalizes $\UL$. The same can be done for the TopKAT formulation. A TopKAT version of the $\UL$ proof system has been already proposed in \cite[Figure 6]{zhang_incorrectness_2022}. The reduction here considered refers to such system, with the following minor differences:
\begin{itemize}
	\item We consider only propositional fragments of the logic, meaning that the rules (assume) and 
	(identity) are replaced by the following single  (transfer) rule:
	\vspace*{-5pt}
	\begin{prooftree}
		\AxiomC{$\mathtt{c} \in \Atom$}
		\RightLabel{(transfer)}
		\UnaryInfC{$\ptriul{a}{c}{a \intrp{c}}$}
	\end{prooftree}
	\item The premises of the (consequence) rule in \cite[Figure 6]{zhang_incorrectness_2022}, $b \leq b'$ and $c' \leq c$, are relaxed to $\top b \leq \top b'$ and $\top c' \leq \top c$. Notice that the former implies the latter. Furthermore, the soundness proof of \cite[Theorem 4]{zhang_incorrectness_2022} is not affected by this change, because
	\( (\top b' \geq \top b \land \top c \geq \top c' \land \top bp \geq c) \Rightarrow \top b' p \geq \top b p \geq \top c \geq \top c' \),
	and, by \cite[Theorem 3]{zhang_incorrectness_2022}, it holds that $\top b'p \geq \top c'$ entails $\top b'p \geq c'$.
	\item  The (limit) rules of $\LCTK_A$  and $\UL$ differ on the distributivity condition. We assume that 
	distributivity also holds in $\UL$. 
\end{itemize}

By instantiating to the trivial abstract domain $A_{tr}$, it turns out that the two proof systems become equivalent. 
\begin{restatable}[$\LCTK_{A_{tr}} \equiv \UL$]{theorem}{thmlcltopulequiv} \label{thm:lctk_ul_equiv}
	Let $K$ be a TopKAT$^*$. For any $a,b \in K$, $\mathtt{t} \in T_{\Sigma, B}$:
	\[ \ptrittr{a}{t}{b} \quad \Leftrightarrow \quad \ptriul{a}{t}{b}. \]
\end{restatable}
\begin{proof}
	The $\Rightarrow$ direction can be shown as in Theorem \ref{thm:lcl_ul_equiv}. We focus instead on the opposite direction.
	
	\noindent $\Leftarrow$) The proof is by induction on the derivation tree of $\ptriul{a}{t}{b}$. The rules of Table \ref{tbl:rules_equiv} are immediate, meaning that only (empty),(choice),(disj),(iterate zero) and (iterate non-zero) are left to prove. As in Theorem \ref{thm:lcl_ul_equiv}, we leverage the completeness of the proof system and likewise, in order to prove that a triple is valid, it is enough to show that condition (1) holds because (2) is trivially true.
	
	\vspace{5px} \noindent (empty)\,: $\vtrittr{a}{t}{0}$ is trivial since $\top \cdot 0 = 0 \leq \top a \intrp{t}$ is always true, therefore by completeness of $\LCTK_{A_{tr}}$, $\ptrittr{a}{t}{0}$.

	\vspace{5px} \noindent (choice)\,: By induction we have, without loss of generality, $\ptrittr{a}{t_1}{b}$ and we need to show $\ptrittr{a}{t_1 + t_2}{b}$. Notice that for any $\mathtt{t_2} \in T_{\Sigma, B}$, $\vtrittr{a}{t_2}{0}$ because $\top \cdot 0 = 0 \leq \top a \intrp{t_2}$, therefore by completeness of $\LCTK_{A_{tr}}$ $\ptrittr{a}{t_2}{0}$. Finally, (join) yields the result $\ptrittr{a}{t_1 + t_2}{b}$.
	
	\vspace{5px} \noindent (disj)\,: By induction we have $\ptrittr{a_1}{t}{b_1}$ and $\ptrittr{a_2}{t}{b_2}$ and we need to show $\ptrittr{a_1 + a_2}{t}{b_1 + b_2}$. By soundness $\top b_1 \leq \top a_1 \intrp{t}$ and $\top b_2 \leq \top a_2 \intrp{t}$, so that $\top (b_1 + b_2) \leq \top (a_1 + a_2) \intrp{t}$. Finally, by completeness we retrieve the result: $\ptrittr{a_1 + a_2}{t}{b_1 + b_2}$.

	\vspace{5px} \noindent (iterate zero)\,: By (\ref{eqn:star_unfold}) $\intrp{t^*} = (\intrp{t})^* \geq 1$, so that $\top a \intrp{t^*} \geq \top a$ and by completeness $\ptrittr{a}{t^*}{a}$.
	
	\vspace{5px} \noindent (iterate non-zero)\,: By (\ref{eqn:star_unfold}) $\intrp{t^*} = (\intrp{t})^* \geq (\intrp{t})^* \intrp{t} = \intrp{t^*} \intrp{t} = \intrp{t^* \cdot t}$. Thus, by hypothesis $\top b \leq \top a \intrp{t^* \cdot t} \leq \top a \intrp{t^*}$ and by completeness $\ptrittr{a}{t^*}{b}$.
	\qed
\end{proof}

In turn, the logical soundness and completeness of $\UL$ can be retrieved as a consequence of those of $\LCTK$.

\begin{restatable}{corollary}{ulsndcompl}\label{cor:ul_snd_compl2}
	Under the same hypotheses of Theorem \ref{thm:lctk_ul_equiv}, the proof system $\UL$ is sound and complete, that is,
	\(\ptriul{p}{t}{q} \; \Leftrightarrow \;\: \vtriul{p}{t}{q}\).
\end{restatable}
\begin{proof}
	\begin{align*}
		\ptriul{p}{t}{q} &\Leftrightarrow \qquad \textrm{[by Theorem \ref{thm:lctk_ul_equiv}]} \\
		\ptrittr{p}{t}{q} &\Leftrightarrow \qquad \textrm{[by Theorems \ref{thm:lsound_top}, \ref{thm:lcompl_top}]} \\
		\vtrittr{p}{t}{q} &\Leftrightarrow \\
		\vtriul{p}{t}{q} & \tag*{\qed}
	\end{align*}
\end{proof}

\section{Incorrectness Logic in TopKAT}\label{IL-TopKAT}
In Section \ref{sect:lctk} we proposed an alternative formulation of local completeness logic based on a TopKAT. Here, we take a similar path as in Section \ref{sect:inc} in adding support for abnormal termination. Like in the modal formulation, each language term is interpreted as a pair of KAT elements, meaning that the interpretation function has signature $\intrp{\cdot} : T_{\Sigma, B} \ra (K \times K)$ and is defined inductively as in (\ref{def:il_intrp}). Following the approach of Section \ref{sect:lctk} we consider an element $\top pa$ as a proxy for the strongest postcondition and in particular we have two such elements corresponding to, resp., normal and abnormal termination: $\top p \intrpok{t}$ and $\top p \intrperr{t}$.

The abstract semantic is a function of type $\asema{\cdot} : T_{\Sigma, B} \ra (A \ra (A \times A))$, where $A$ is a top Kleene abstract domain. As usual it is defined inductively, with base case $\mathtt{a} \in \Atom$ defined as
\[ \asema{a} p^\sharp \triangleq \ilc{\alpha(\gamma(p^\sharp) \intrpok{a})}{\alpha(\gamma(p^\sharp) \intrperr{a})} \]
and the inductive cases as in (\ref{def:abs_sem_il}). It turns out that if $K$ is a TopKAT$^*$ then the abstract semantic is sound and monotone.

\subsection{Local completeness incorrectness logic on TopKAT}
The proof system $\LCIL$ can be adapted to the TopKAT framework as follows. Let $A$ be a top Kleene abstract domain on a TopKAT$^*$ $K$ and a language $T_{\Sigma, B}$. If $a,b \in K$ and $\mathtt{t} \in T_{\Sigma, B}$, then 
\begin{itemize}
	\item The triple $\triok{a}{t}{b}$ is valid (1) $\top b \leq \top a \intrpok{t}$, and (2) $\asemaok{t} \alpha(\top a) = \alpha(\top b) = \alpha(\top a \intrpok{t})$.
	\item The triple $\trierr{a}{t}{c}$ is valid (1) $\top c \leq \top a \intrperr{t}$, and (2) $\asemaerr{t} \alpha(\top a) = \alpha(\top c) = \alpha(\top a \intrperr{t})$.
	\item A triple $\triokerr{a}{t}{b}{c}$ is valid if both $\triok{a}{t}{b}$ and $\trierr{a}{t}{c}$ are valid. In particular, if $b=c$ then the triple $\triboth{a}{t}{b}$ is valid.
\end{itemize}

The proof system, denoted $\LCTIL_A$, has the same rules as $\LCIL_A$ in Figure \ref{fig:ilcl} except for the rules (transfer) and (relax), modified as follows:
\begin{center}
	\begin{prooftree}
		\AxiomC{$\mathtt{c} \in \Atom$}
		\AxiomC{$\mathbb{C}^A_a(\intrpok{c})$}
		\AxiomC{$\mathbb{C}^A_a(\intrperr{c})$}
		\RightLabel{(transfer)}
		\TrinaryInfC{$\ptritokerr{a}{c}{a \intrpok{c}}{a \intrperr{c}}$}
	\end{prooftree}
	\begin{prooftree}
		\AxiomC{$\top a' \leq \top a \leq A(\top a') \quad \ptritboth{a'}{t}{b'} \quad \top b \leq \top b' \leq A(\top b)$}
		\RightLabel{(relax)}
		\UnaryInfC{$\ptritboth{a}{t}{b}$}
	\end{prooftree}
\end{center}
Moreover, for the rule (limit) we assume that the infinitary disjunction exists and that $\top$ distributes over $\disj a_n$. Observe that these additional conditions are the same as (\ref{tk1}) for the TopKAT formulation of the local completeness logic.

It turn out that $\LCTIL_A$ is logically sound and complete.

\begin{theorem}[Logical Soundness of $\LCTIL_A$]
	The triples provable in $\LCTIL_A$ are valid.
\end{theorem}
\begin{proof}
	Analogous to the proof of Theorem \ref{thm:lsound_il}. \qed
\end{proof}

\begin{theorem}[Logical completeness of $\LCTIL_A$]
	Assume that the atoms in $\mathtt{t} \in T_{\Sigma, B}$ are globally complete, i.e., for all $\mathtt{a} \in \Atom$, $\mathbb{C}^A(\intrpok{a})$ and $\mathbb{C}^A(\intrperr{a})$. If $\triokerr{a}{t}{b}{c}$ is valid,then it is provable in $\LCTIL_A$.
\end{theorem}
\begin{proof}
	It is easy to check that $\asemaok{t} \alpha(\top a) = \alpha(\top a \intrpok{t})$ and $\asemaerr{t} \alpha(\top a) = \alpha(\top a \intrperr{t})$, as in Lemmas \ref{lemma:glob_compl_il} and \ref{lemma:glob_compl_top}. Then we can proceed as in Theorem \ref{thm:lcompl_il} to obtain the result. \qed
\end{proof}

\subsection{Relationship with Incorrectness logic}
In Section \ref{sect:lcil_il} and \ref{sect:lctk_ul} we have shown that the modal formulation of $\LCIL$ generalizes $\IL$ and $\LCTK$ generalizes $\UL$. Likewise, it is possible to show that $\LCTIL$ generalizes $\IL$. A TopKAT formulation of the $\IL$ proof system was already proposed in \cite[Figure 4]{zhang_incorrectness_2022}. We refer to such system for the proposed reduction, with the same minor changes that we discussed in Section \ref{sect:lctk_ul}, namely: 1) we consider only propositional fragments of the logic, 2) the premises of (consequence) are relaxed to $\top b \leq \top b'$ and $\top c' \leq \top c$, 3) we assume that the distributivity condition is satisfied also in $\IL$.

Similarly to the reductions already proposed in this paper, we let $A = A_{tr}$, where $A_{tr}$ is the trivial abstract domain.

\begin{theorem}[$\LCTIL \equiv \IL$] \label{thm:lctil_il_equiv}
	Let $K$ be a TopKAT$^*$. For any $a,b \in K$, $\mathtt{t} \in T_{\Sigma, B}$,
	\[ \ptritokerrtr{a}{t}{b}{c} \quad \Leftrightarrow \quad \ptriokerril{a}{t}{b}{c} \]
\end{theorem}
\begin{proof}
	Analogous to the proof of Theorem \ref{thm:lcil_il_equiv}. \qed
\end{proof}

\noindent Finally, we can retrieve the logical soundness and completeness of $\IL$ as a consequence of those of $\LCTIL$.

\begin{corollary}
	Under the hypothesis of Theorem \ref{thm:lctil_il_equiv} the proof system $\IL$ is logically sound and complete, that is, $\ptriokerril{a}{t}{b}{c} \; \Leftrightarrow \; \vtriokerril{a}{t}{b}{c}$.
\end{corollary}

\section{Conclusion}

This work has shown that the abstract interpretation-based local completeness  logic introduced in \cite{bruni_logic_2021} can be generalized to and interpreted
in Kleene algebra with tests. In particular, we proved that this can be achieved both for KATs extended 
with a modal backward diamond operator playing the role of strongest postcondition, 
 and for KATs endowed with a
top element. 
Our results generalize both 
the modal \cite{fahrenberg_algebra_2021} 
and top \cite{zhang_incorrectness_2022} KAT
approaches that encode Hoare correctness and O'Hearn incorrectness logic using different classes of KATs. 
In particular, our KAT-based logic leverages an abstract interpretation 
of KAT, a problem that was not studied so far. 

Our plan for future work includes, but is not limited to, the following questions. 
\begin{itemize}
\item For a KAT with top $\top$, following the technical idea underlying the approach 
by Zhang et al.~\cite{zhang_incorrectness_2022}, 
we defined an abstract domain as an approximation of all
the algebraic elements of type $\top\cdot a$, where $a$ is any element of the KAT (cf.\ Definition~\ref{def:tkad}). Although this definition technically works, it is somehow artificial, 
because the elements $\top \cdot a$ do not carry a clear intuitive meaning. As an interesting future task, we would like 
to characterize under which conditions an element $\top \cdot a$ coincides with $\top \cdot p$ for some test $p \in \Test(K)$, 
and if such test $p$ is unique.
\item This work is a first step towards an \emph{algebraic and equational approach to abstract interpretation}. 
We envisage that the reasoning made by an abstract interpreter of programs could be made purely equational within a KAT equipped
with a suitable collection of axioms. The ambition would be to conceive a notion of \emph{abstract Kleene algebra} (AKA) 
making this slogan true: \emph{AKA is for the abstract interpretation of programs what KAT is for concrete interpretation of programs}. 
\end{itemize}

\subsubsection*{Acknowledgements.}
Francesco Ranzato has been partially funded by  the \emph{Italian Ministry of University and Research}, under the PRIN 2017 project no.\ 201784YSZ5 ``AnalysiS of PRogram Analyses (ASPRA)'', 
by \emph{Facebook Research}, under a ``Probability and Programming Research Award'', and 
by an \emph{Amazon Research Award} for ``AWS Automated Reasoning''.

\end{document}